\documentclass[a4paper,onecolumn,11pt,accepted=2022-07-11]{quantumarticle}
\pdfoutput=1
\usepackage[utf8]{inputenc}
\usepackage[english]{babel}
\usepackage[T1]{fontenc}

\usepackage{fullpage}
\usepackage{mdframed}

\usepackage[numbers,sort&compress]{natbib}

\usepackage{amsmath, amssymb, amsthm,verbatim, graphicx, xcolor,bbm}
\usepackage{stmaryrd} 
\usepackage{url}

\usepackage{mathrsfs}  

\usepackage{braket}

\makeatletter
\newtheorem*{rep@theorem}{\rep@title}
\newcommand{\newreptheorem}[2]{%
\newenvironment{rep#1}[1]{%
 \def\rep@title{#2 \ref{##1}}%
 \begin{rep@theorem}}%
 {\end{rep@theorem}}}
\makeatother

\newtheorem{theo}{Theorem} 
\newreptheorem{theo}{Theorem}

\newtheorem{lemma}[theo]{Lemma}

\newtheorem{corol}[theo]{Corollary}
\newtheorem{defn}[theo]{Definition}
\newtheorem{rem}[theo]{Remark}
\newtheorem{claim}[theo]{Claim}

\def\U{\mathbf{U}}

\def\C{\mathbbm{C}}

\def\F{\mathbbm{F}}

\def\uX{\mathcal{X}}
\def\uY{\mathcal{Y}}
\def\uZ{\mathcal{Z}}

\def\cA{\mathcal{A}}
\def\cB{\mathcal{B}}
\def\cC{\mathcal{C}}
\def\cD{\mathcal{D}}
\def\P{\mathcal{P}}
\def\cE{\mathcal{E}}

\def\I{\mathcal{I}}

\def\cC{\mathcal{C}}

\def\1{\mathbbm{1}}

\DeclareUnicodeCharacter{00A0}{ }

\usepackage{hyperref}

\begin{document}

\title{Quantum XYZ Product Codes}

\author{Anthony Leverrier}
\affiliation{Inria, France}
\orcid{0000-0002-6707-1458}
\email{anthony.leverrier@inria.fr}

\author{Simon Apers}
\affiliation{CNRS, IRIF, Universit\'e Paris Cit\'e}
\email{smgapers@gmail.com}
\orcid{0000-0003-3823-6804}

\author{Christophe Vuillot}
\affiliation{Inria, France}
\orcid{0000-0002-3445-0179}
\email{christophe.vuillot@inria.fr}

\maketitle

\begin{abstract}
We study a three-fold variant of the hypergraph product code construction, differing from the standard homological product of three classical codes. When instantiated with 3 classical LDPC codes, this ``XYZ product'' yields a \emph{non CSS} quantum LDPC code which might display a large minimum distance. 
The simplest instance of this construction, corresponding to the product of 3 repetition codes, is a non CSS variant of the 3-dimensional toric code known as the Chamon code. 
The general construction was introduced in Denise Maurice's PhD thesis, but has remained poorly understood so far. 
The reason is that while hypergraph product codes can be analyzed with combinatorial tools, the XYZ product codes also depend crucially on the algebraic properties of the parity-check matrices of the three classical codes, making their analysis much more involved. 

Our main motivation for studying XYZ product codes is that the natural representatives of logical operators are two-dimensional objects. This contrasts with standard hypergraph product codes in 3 dimensions which always admit one-dimensional logical operators. In particular, specific instances of XYZ product codes with constant rate \emph{might} display a minimum distance as large as $\Theta(N^{2/3})$. While we do not prove this result here, we obtain the dimension of a large class of XYZ product codes, and when restricting to codes with dimension 1, we reduce the problem of computing the minimum distance to a more elementary combinatorial problem involving binary 3-tensors. 
We also discuss in detail some families of XYZ product codes that can be embedded in three dimensions with local interaction. Some of these codes seem to share properties with Haah's cubic codes and might be interesting candidates for self-correcting quantum memories with a logarithmic energy barrier. 
\end{abstract}


\newpage

\section{Introduction and summary}
\label{sec:intro}

There is a large interest in quantum low-density parity-check (LDPC) codes, both for their relevance for building fault-tolerant quantum computers and for their nontrivial topological properties.
Despite this, many questions are still open: Can quantum LDPC codes have linear minimum distance? Do there exist quantum LDPC codes that are locally testable? Can quantum LDPC codes be embedded in 3 \emph{physical} dimensions and describe self-correcting quantum memories?
We study these questions for a family of quantum LDPC codes called \textit{quantum XYZ product codes}, which are a non-CSS, 3-fold variant of the hypergraph product code, derived from a triple of classical codes with parity-check matrices $H_1,H_2,H_3$.
They were initially described in \cite{mau14} and as a special case include the Chamon code \cite{cha05,BLT11}.
We conjecture that this family includes codes that have constant rate and minimum distance $\Theta(N^{2/3})$.
As a step towards this conjecture we show that, under certain conditions, the dimension and the distance of the code are determined by the \textit{tensor Sylvester equation} (over $\F_2$):
\[
(H_1H_1^T \otimes \1 \otimes \1) X
= (\1 \otimes H_2 H_2^T \otimes \1) X
= (\1 \otimes \1 \otimes H_3 H_3^T) X + R.
\]
For $R = 0$ this equation characterizes the code dimension.
For $R = R(H_1,H_2,H_3) \neq 0$ it characterizes the minimum distance.
More precisely, the distance is given by how closely (in Hamming weight) the equation can be satisfied.
This observation stresses the algebraic nature of the code family.
A particular subclass of interest are the \textit{cyclic} XYZ product codes where $H_1,H_2,H_3$ are circulant matrices.
We study the simplest code in this family (beyond the Chamon code), for which we show the distance is likely upper bounded by $\sqrt{N}$. We also show the existence of fractal operators (as in Haah's cubic code \cite{haa11}) for cyclic XYZ product codes. Such operators exclude local testability.

\subsection{Quantum LDPC codes}

While specific classes of quantum LDPC codes such as surface codes \cite{kit03} are rather well understood, they suffer from inherent limitations of their parameters: for instance their minimum distance is upper bounded by the square root of their length \cite{BT09}. In fact, beating this bound turned out to be very challenging, even for general LDPC codes that need not be embeddable on a finite-dimensional lattice. 
More precisely, an old construction due to Freedman, Meyer and Luo gives a minimum distance $\Theta (N^{1/2} \log^{1/4} N)$ \cite{FML02}, and only very recently new constructions based on high-dimensional expanders have brought polylogarithmic improvements \cite{EKZ20,KT20}.
These were holding the record until the recent breakthrough of Ref.~\cite{HHO20} which introduced \emph{fiber bundle codes}, a variant of the hypergraph product codes relying on a \emph{twisted} homological product.
A particular instance of fiber bundle codes was shown to display a minimum distance of $\Theta(N^{3/5})$, up to polylogarithmic factors. See also \cite{BE20} for an explicit construction with the same parameters. Even more recently, Panteleev and Kalachev showed that a generalization of fiber bundle codes, named \emph{lifted product codes} \cite{PK20} (and originally introduced in Ref.~\cite{PK19} in a more restrictive version) have an almost linear minimum distance, $d = \Theta(n/\log n)$, with a logarithmic dimension, $k= \Theta(\log n)$. Combining this construction with the balancing trick of Hastings \cite{has17b}, one obtains codes such that $kd^2 = \Theta(n^2)$. See Ref.~\cite{BE21} for a recent review of the current state-of-the-art on quantum LDPC codes.

It is a daunting open question to understand whether there exist quantum LDPC codes with both a large rate and minimum distance.
Recall that in the classical setting, LDPC codes with constant rate and linear minimum distance are easy to obtain, for instance by picking a sparse parity-check matrix at random.
Expansion properties of the resulting code will exclude the existence of low-weight operators with zero syndrome.
The quantum case is significantly more complex, in particular because we cannot simply exclude low-weight operators with zero syndrome.
Indeed, any quantum LDPC code will have such operators (corresponding to the generators of the code).
Rather, the challenge is to exclude low-weight \emph{logical} operators (\textit{i.e.}, operators that send one codeword to a different one). 
Recall that a quantum stabilizer code \cite{got97} of length $N$ is defined as the common $+1$ eigenspace of a set $\{g_1, \ldots, g_m\}$ of commuting Pauli operators acting on $(\C^2)^{\otimes N}$, that is, the space 
\[ \mathcal{Q} := \mathrm{span} \{ |\psi\rangle \in  (\C^2)^{\otimes N} \: : \: g_i |\psi\rangle = |\psi\rangle, \forall i \in [m] \}.\]
Such a code is said to be \emph{LDPC} if all the generators $g_i$ act nontrivially on at most a constant number of qubits, and if each qubit is only acted on by a constant number of generators.
With this language, a quantum LDPC stabilizer code also corresponds to the ground space of the local Hamiltonian $H = \frac{1}{m} \sum_{i=1}^m \Pi_i$, with $\Pi_i = \frac{1}{2}(\mathbbm{1}- g_i)$.\\

If we denote by $\mathcal{S} := \langle g_1, \ldots, g_m \rangle$ the stabilizer group of the code, then the logical operators correspond to the set $N(\mathcal{S})\setminus \mathcal{S}$ where ${N}(\mathcal{S})$ is the normalizer of $\mathcal{S}$. 
The code dimension is $k = |N(\mathcal{S})/\mathcal{S}|$ and the minimum distance $d_{\min}$ of the code is given by the minimum weight (\textit{i.e.} the number of qubits on which it acts nontrivially) of an operator that leaves the code space globally invariant, but which is not in the stabilizer group:
\begin{align}
d_{\min} = \min \{ |w| \: : \: w \in N(\mathcal{S})\setminus \mathcal{S}\}.
\end{align}
This quantity can be particularly challenging to compute for quantum LDPC codes since most low-weight operators in $N(\mathcal{S})$ effectively belong to the stabilizer group $\mathcal{S}$ and therefore do not affect the distance. 
Quantum CSS codes offer a much simplified setup: these codes correspond to stabilizer codes for which the generators are either products of Pauli-$X$ operators $\sigma_1 = \left[ \begin{smallmatrix} 0 & 1 \\ 1& 0\end{smallmatrix}\right]$ and the identity $\1$, or products of Pauli-$Z$ operators $\sigma_3 = \left[ \begin{smallmatrix} 1 & 0 \\ 0& -1\end{smallmatrix}\right]$ and the identity $\1$ \cite{CS96,ste96b,ste96}. These families are easier to study because the commutation relations required to make the stabilizer Abelian simply need to be checked between $X$-type and $Z$-type generators. In particular, such quantum codes can be described by a pair of classical codes.
\begin{defn}[CSS code]
A quantum CSS code with parameters $\llbracket N, k, d_{\min} \rrbracket$ is a
pair of classical codes $\cC_X, \cC_Z \subseteq \mathbbm{F}_2^N$ such that
$\cC_X^\perp \subseteq \cC_Z$, or equivalently $\cC_Z^\perp \subseteq \cC_X$.
The code space corresponds to the linear span of $\left\{ \sum_{z \in \cC_Z^\perp} |x+z\rangle
\: : \: x\in \cC_X\right\}$, where $\left\{ |x\rangle\: : \: x\in \mathbbm{F}_2^N\right\}$ is the canonical basis of $(\mathbbm{C}^2)^{\otimes N}$.
\end{defn}
\noindent
The dimension of the code is given by $\dim(\cC_X/\cC_Z^\perp) =
\dim\cC_X + \dim\cC_Z - N$ and its minimum distance is
given by 
\[d_{\min}^{(\text{CSS})} = \min (d_X, d_Z),\]
 with $d_X = \min\{ |w| \: : \: w \in \cC_X
\setminus \cC_Z^\perp\}$ and $d_Z = \min\{ |w|\: : \: w \in \cC_Z \setminus \cC_X^\perp\}$. Here, $|w|$ stands for the Hamming weight of the word $w$.
All the code constructions with large distance mentioned above are instances of CSS codes \cite{FML02,KKL16,EKZ20,KT20,HHO20}.
An appealing property of CSS codes is that they can alternatively be described using chain complexes.
A chain complex (of length 3) is described by 3 binary vector spaces $C_0, C_1, C_2$ together with \emph{boundary operators} $\partial_1, \partial_2$, and is depicted as
\[
C_2  \xrightarrow{\partial_{2}} C_1 \xrightarrow{\partial_1} C_0.
\]
The boundary operators are such that $\partial_1 \partial_2 = 0$.
In this case, the classical codes $\cC_X$ and $\cC_Z$ defining the CSS code are given by 
\[
\cC_X
= \ker \partial_1 \quad \text{and} \quad
\cC_Z
= (\mathrm{Im}\, \partial_{2})^\perp = \ker \partial_2^T.
\]
The condition that $\cC_Z^\perp \subseteq \cC_X$ now naturally follows from the condition that $\partial_1  \partial_2 = 0$.
The qubits correspond to the space $C_1$, while the $X$- and $Z$-check operators correspond to $C_0$ and $C_2$, respectively.\\

One might intuitively expect that non-CSS codes offer better prospects in terms of possible minimum distance, albeit at the cost of making their analysis more involved.
It turns out that this is not the case, as was observed by Ref.~\cite{BTL10}: starting from any non CSS code of parameters $\llbracket n,k,d\rrbracket $, one can construct a CSS code with essentially the same asymptotic parameters\footnote{The hyperbicycle codes of \cite{KP13} give an example where the non CSS version can display slightly better parameters than the standard CSS construction of hypergraph product codes \cite{TZ13}.}:
\[
\llbracket n_\text{CSS}
= 4n,k_\text{CSS}=2k,d_\text{CSS}=2d \rrbracket,
\]
and generators of weight at most doubled compared to the original code (we detail this construction in Section \ref{sub:CSS}).
This fact, together with their \textit{a priori} more involved analysis, explains why non CSS codes have been mostly ignored as an approach to devise quantum LDPC codes with large distance\footnote{We note, however, that non CSS codes are investigated in the context of quantum fault-tolerance where their non CSS nature offers interesting properties in terms of resistance to biased noise for instance \cite{BTB20}. For instance, the XZZX code of Ref.~\cite{BTB20} is a 2-dimensional version of the Chamon code.}.
A notable exception is the \emph{Chamon code} \cite{cha05} which was initially introduced as a three-dimensional quantum Hamiltonian with intriguing topological properties and later analyzed in depth by Bravyi, Leemhuis and Terhal \cite{BLT11}. This code shares with the 3D toric code the fact that it can be embedded in three dimensions with local interaction: the generators have weight 6 and act on their two neighboring qubits in each spatial direction with either $X$-, $Y$-, or $Z$-type Pauli operators. 
Interestingly, while the 3D toric code admits string-like logical operators (and thus has a minimum distance upper bounded by $N^{1/3}$), the natural logical operators of the Chamon code are membrane-like, giving some hope that the minimum distance could be much larger.
Unfortunately, the Chamon code is known to have distance only $O(\sqrt{N})$. 
This follows from the fact that certain ``flexible'' string operators of the Chamon code live in a plane orthogonal to the $[1,1,1]^T$ direction of the lattice and that the topology of this plane is identical to a 2-dimensional toric code, hence there exists a  noncontractible loop (corresponding to a logical operator) of length at most $\sqrt{N}$.
Whether this upper bound is tight or not is still an open question, as far as we know. Note that this upper bound only applies to the \emph{stabilizer version} of the Chamon code. It also makes sense to study \emph{subsystem code} versions where not all logical qubits are used to encode information, with the hope that the chosen qubits necessarily have large weight logical operators. While \cite{BLT11} shows that the subsystem code version still display a $\sqrt{N}$ upper bound for the minimum distance for some choice of parameters, it is not excluded that better choices would lead to distance as large as $\Theta(N^{2/3})$.
An important aspect of the analysis of \cite{BLT11} is that the properties of the Chamon code have an intrinsic algebraic nature, suggesting that the code is very different from its toric code cousin.
This shows up very clearly when computing the quantum dimension of the Chamon code:
if the code is defined on a lattice of size $n_1 \times n_2 \times n_3$ then its dimension is
\[
k_{\text{Chamon}}
= 4 \gcd(n_1, n_2, n_3).
\]
Recall on the other hand that the dimension of topological codes such as the surface code only depend on the genus of the surface (number of holes) and are independent of the lattice size.
The algebraic properties of the parity-check matrix essentially show up because the product of the three Pauli matrices is proportional to the identity: $\sigma_1 \sigma_2 \sigma_3 = i\1_2$, with $\sigma_2 = \left[ \begin{smallmatrix} 0 & -i \\ i& 0\end{smallmatrix}\right]$.
Because of this, nontrivial relations might occur among the code generators, and this is essentially what $k_{\text{Chamon}} $ above counts since the Chamon code has as many generators as qubits.
This is in contrast to typical CSS codes, where the generators involve only $X$- and $Z$-type Pauli operators.

\subsection{Hypergraph Product codes and generalizations}

\emph{Hypergraph product codes} \cite{TZ13} are a constant-rate generalization of the toric code obtained by taking the (homological) product of two classical codes.
In particular, the product of two good classical LDPC codes yields a quantum code with constant rate and minimum distance $\Theta(\sqrt{N})$. At the moment, no quantum LDPC code family is known to beat this bound while displaying a constant rate. 
Such codes could be useful to perform quantum computation in a fault-tolerant manner with a constant overhead in terms of qubits \cite{FGL18}.
Hypergraph product codes can be understood as a special instance of the homological product of two chain complexes of length two \cite{FH14, BH14}.
Concretely, consider two classical codes $\cC_1 = \ker H_1$ and $\cC_2 = \ker H_2$ described by their parity-check matrices $H_i$ of size $m_i \times n_i$ for $i =1,2$.
The hypergraph product code associated to these codes is given by the quantum CSS code associated to the 3-complex:
\[
\F_2^{n_1 \times n_2}
\xrightarrow{(H_1 \otimes \1 , \1 \otimes H_2) } (\F_2^{m_1 \times n_2} \oplus \F_2^{n_1 \times m_2} )
\xrightarrow{1 \otimes H_2 + H_1 \otimes \1} \F_2^{m_1 \times m_2}.
\]
We note that by considering the homological product of two random complexes of length 3 (\textit{i.e.}, random CSS codes), Bravyi and Hastings \cite{BH14} showed the existence of CSS codes with linear minimum distance.
Unfortunately the construction requires generators of weight $\Theta(\sqrt{N})$ and therefore do not satisfy the LDPC condition.\\

The 2-dimensional toric code is a special instance of the hypergraph product code construction corresponding to the product of two repetition codes. It is straightforward to generalize the hypergraph product code construction to $D$-fold homological products  \cite{C19,ZP19,QVR20}, and the special case of the repetition code will naturally yield the $D$-dimensional toric code.

\subsection{XYZ product codes}
In her PhD thesis, Denise Maurice noticed that one could also consider a different generalization of the hypergraph product code that would yield a non CSS code when applied to 3 classical codes \cite{mau14}.
The idea is to associate a different code to each of the $X$-, $Y$- or $Z$-Pauli checks (see Fig.~\ref{fig:cubecube-intro}).
She noted that this XYZ product yields the Chamon code when the three classical codes are repetition codes\footnote{Ref.~\cite{BLT11} also suggested some generalizations of the Chamon code replacing the repetition code by symmetric circulant matrices with even row weight. We will consider a closely related construction in Section \ref{sec:cyclic} but with circulant matrices of odd weight.}.
She also gave lower bounds on the dimension of general XYZ product codes, but warned that these were not tight because they ignored potential relations between the generators caused by $XYZ=i\1$.
Consequently, she could not give non trivial lower bounds on the minimum distance, because this would at least require to know how many logical operators exist.
Finally, she established a general upper bound on the minimum distance of the form $O(N^{2/3})$.
\begin{figure}[!h]
\centering
\includegraphics[width=0.4\linewidth]{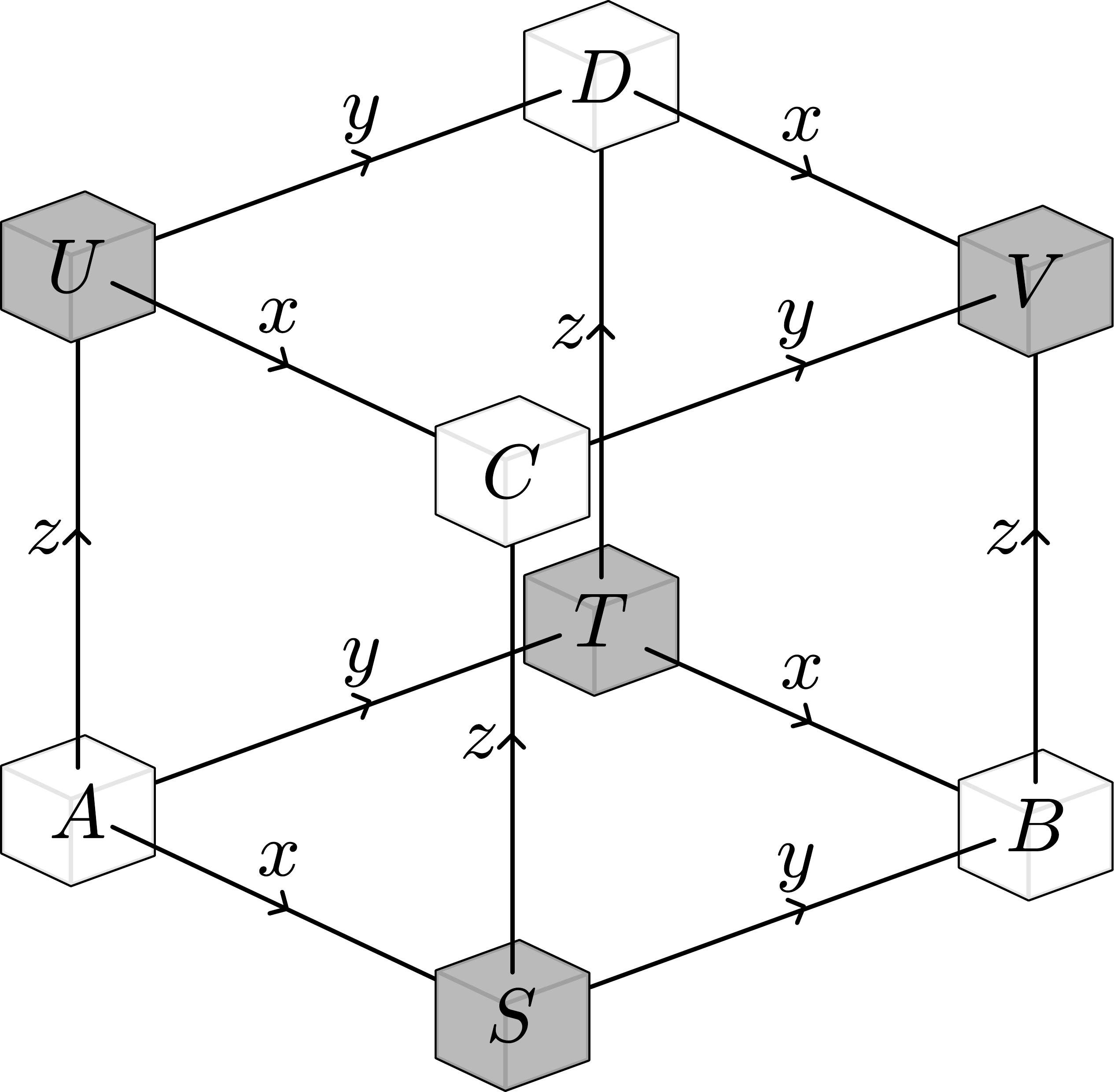}
\caption{3D structure of the XYZ product code: qubits are indexed by the 3-tensors $A, B, C, D$, generators by $S, T, U, V$. The $x$, $y$ and $z$ labels are shorthand for the operators $H_1 \otimes \1 \otimes \1$, $\1 \otimes H_2 \otimes \1$ and $\1 \otimes \1 \otimes H_3$, respectively.}  
\label{fig:cubecube-intro}
\end{figure}

In addition to difficulties to compute the quantum dimension of XYZ product codes, their non CSS nature also makes the understanding of the minimum distance more complex: in particular, it is not possible to restrict the analysis to $X$-type or $Z$-type logical operators similarly to the CSS case, and one must \textit{a priori} consider all possible logical operators. We discuss in Section \ref{sub:CSS} a possible CSS variant of the XYZ product codes that could be better suited for a distance analysis. However, this CSS version involves 4 times as many qubits and also a lot of new generators, and consequently turns out not be much easier to study.
The interest behind this code construction is however that it might break the $\sqrt{N}$ bound for the minimum distance.
The reason for which hypergraph product codes have a $\sqrt{N}$ distance is that there exist logical operators that basically correspond to codewords of one of the two classical codes.
With the XYZ product, the hope is that the third code acting with $Y$-Pauli operators along the third coordinate will force low-weight logical operators to be a product of two out of the three classical codes.
Indeed, the natural logical operators have exactly this form.
The problem then is to understand whether there exist equivalent logical operators (differing by an element of the stabilizer group) with a much lower weight. This is for instance the case for the Chamon code since the minimum distance is only $O(N^{1/2})$.\\

More formally, the XYZ product code construction is a generalization of the hypergraph product code involving a third classical code which will enforce Pauli $Y$-type constraints.
Given three parity-check matrices $H_1$, $H_2$ and $H_3$, we first define a chain complex, depicted in Figure \ref{fig:chaincomplex-intro2}.

\begin{figure}[!h]
\centering
\includegraphics[width=0.65\linewidth]{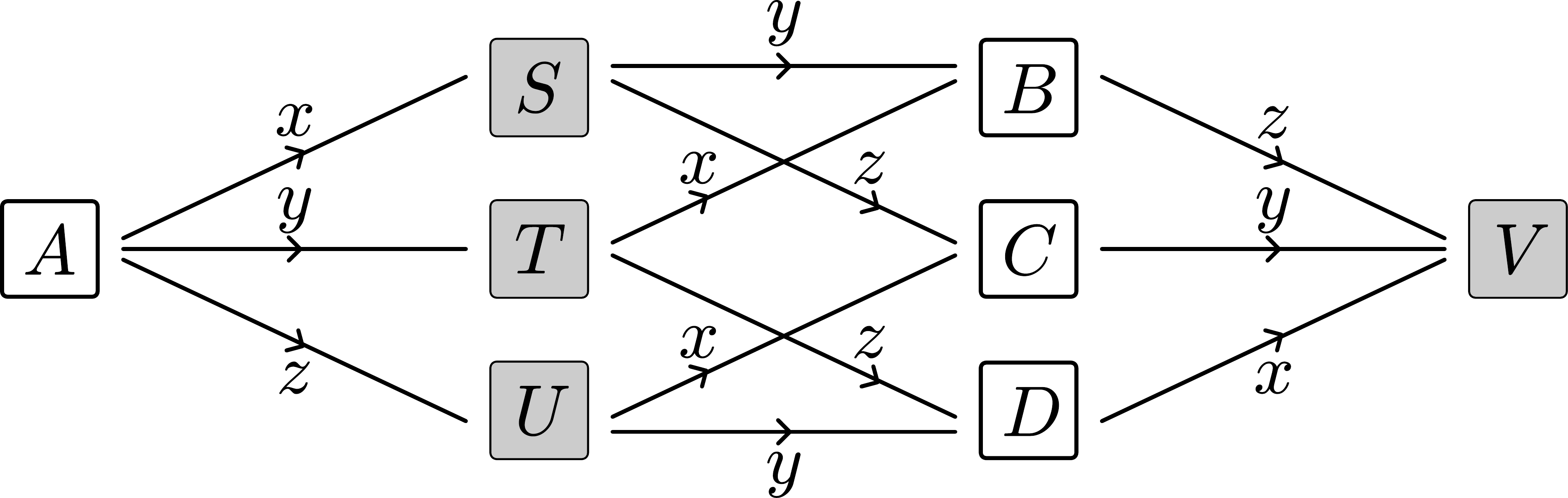}
\caption{\emph{Chain complex} representation of the XYZ product codes.
The $x$, $y$ and $z$ labels are shorthand for the operators $H_1 \otimes \1 \otimes \1$, $\1 \otimes H_2 \otimes \1$ and $\1 \otimes \1 \otimes H_3$, respectively.}  
\label{fig:chaincomplex-intro2}
\end{figure}

We proceed similarly to the hypergraph product code construction.
We identify the resulting vector with 3-tensors
\[
A \in \F_2^{n_1 \times n_2 \times n_3}, \quad
B \in \F_2^{n_1 \times m_2 \times m_3}, \quad
C \in \F_2^{m_1 \times n_2 \times m_3}, \quad
D \in \F_2^{m_1 \times m_2 \times n_3},
\]
which will index the qubits, and 3-tensors
\[
S \in \F_2^{m_1 \times n_2 \times n_3}, \quad
T \in \F_2^{n_1 \times m_2 \times n_3}, \quad
U \in \F_2^{n_1 \times n_2 \times m_3}, \quad
V \in \F_2^{m_1 \times m_2 \times m_3},
\]
which will index the checks.
The total number of qubits is
\begin{align}
\label{eqn:length}
N = \underbrace{n_1 n_2 n_3}_{n_A} +\underbrace{ n_1 m_2 m_3}_{n_B} +\underbrace{ m_1 n_2 m_3}_{n_C} +\underbrace{ m_1 m_2 n_3}_{n_D}.
\end{align}
A general stabilizer element is a tensor over the Pauli group written abstractly as
\begin{align}\label{eqn:gen}
\begin{bmatrix}
{\cal A} \\ {\cal B} \\ {\cal C} \\ {\cal D} \end{bmatrix} = 
\begin{bmatrix}
\uX^\dagger & \uY^\dagger & \uZ^\dagger & \cdot\\
\uY & \uX &\cdot& \uZ^\dagger\\
\uZ & \cdot& \uX & \uY^\dagger\\
\cdot& \uZ & \uY  &\uX^\dagger \end{bmatrix}
\begin{bmatrix} S\\T \\U \\V\end{bmatrix}
=: \Gamma(S,T,U,V).
\end{align}
Here we introduced the operators\footnote{The dimension of the identity operator $\1$ will vary depending on the context.}
\[
\uX = \sigma_1^{H_1 \otimes \1 \otimes \1}, \qquad
\uY = \sigma_2^{\1 \otimes H_2 \otimes \1}, \qquad
\uZ = \sigma_3^{\1 \otimes \1 \otimes H_3}.
\]
We use the convention that an adjoint $\uX^\dag$ denotes the operator $\uX^\dag = \sigma_1^{H_1^T \otimes \1 \otimes \1}$, a product $\uY S$ denotes $\sigma_2^{(\1 \otimes H_2 \otimes \1) S}$, and a sum $\uY S + \uX T$ denotes the \emph{product} $\sigma_2^{(\1 \otimes H_2 \otimes \1) S} \sigma_1^{(H_1 \otimes \1 \otimes \1) T}$.
We call the resulting stabilizer code the \textit{XYZ product code}, and we refer to it by the shorthand ${\cal Q}(H_1,H_2,H_3)$.

\subsection{Our results}

We initiate a systematic study of the XYZ product codes and establish a number of results. 
Unfortunately, we do not manage to establish the code dimension in the most general case, nor do we succeed in proving record-breaking lower bounds on the minimum distance.
In fact, even proving a bound of the form $\omega(N^{1/3})$ for the minimum distance appears to be challenging.
We note that Ref.~\cite{BLT11} claims that such a result for the Chamon code would appear in a future work of one of the authors, but has not as far as we know. 

\subsubsection{Validity of the XYZ product code construction}
A stabilizer code is characterized by its stabilizer group $\mathcal{S} = \langle g_1, \ldots, g_m\rangle$ that should satisfy two properties: it should be Abelian, meaning that any pair of generators $g_i, g_j \in \mathcal{P}^{\otimes N}$ should commute, $[g_i,g_j]=0$, and it should not contain $-\1$. Note that indeed, $- \1 \in \mathcal{S}$ immediately implies that any codeword $|\psi\rangle$ should be stabilized by $-\1$ and therefore $|\psi\rangle = -|\psi\rangle = 0$. 
We prove in Lemma \ref{lem:xyz} that the stabilizer group of the XYZ product codes is Abelian, and this follows directly from the fact that the construction is a 3-fold version of the hypergraph product code construction. 
Proving that the stabilizer group does not contain $-\1$ turns out to be much more challenging, however, and this is already a nontrivial result in the case of the Chamon code \cite{BLT11}. While this question does not show up for CSS construction since a product of $\sigma_1$ or $\sigma_3$ commuting generators can never pick a nontrivial phase, this is no longer the case for general stabilizers involving $\sigma_2$ operators since for instance $\sigma_1 \sigma_2 \sigma_3 = i \1$.

While we do not prove that $-\1 \notin \mathcal{S}$ in the general case of the XYZ product of three arbitrary parity-check matrices $H_1, H_2, H_3$, we establish this fact for a special case of interest where the code dimension is 1 (see Lemma \ref{lem:-1forT}). The corresponding codes are the ones we focus on in most of the paper since they offer the best hope for displaying a large minimum distance.

Moreover we show that in the case where $-\1\in\mathcal{S}$ there is always a way to choose a well defined stabilizer code from $\mathcal{S}$.
This stabilizer code is given by making consistent choices for each elements of $\mathcal{S}$ of whether the code space should belong to its $+1$- or $-1$-eigenspace.
The parameters $\llbracket n,k,d\rrbracket$ of the stabilizer code obtained are independent of the choices for the signs, see Section \ref{sub:minusone}.
We also show that using a standard transformation due to \cite{BTL10}, one can exhibit a CSS version of the XYZ product code which is readily well defined whether or not $-\1\in\mathcal{S}$ and has parameters $\llbracket 4n,2k,2d\rrbracket$, see Section \ref{sub:CSS}.

Our result concerning the dimension in Section~\ref{sec:dim} should be interpreted as usual in the cases where $-\1\not\in\mathcal{S}$.
In a case where $-\1\in\mathcal{S}$ then our results would apply to the well defined stabilizer code resulting from any consistent choice of signs for the generators of the XYZ product construction.

\subsubsection{Dimension of the XYZ product codes}

Consider parity-check matrices $H_i$ of size $m_i \times n_i$ for $i \in \{1,2,3\}$. The dimension of the resulting XYZ product code $\mathcal{Q}(H_1, H_2,H_3)$ is given by the number of physical qubits minus the number of independent generators. If $r$ is the number of ``relations'' between the generators, then the dimension $k$ of the code is given by
\begin{align*}
k
&= (n_1 - m_1)(n_2 - m_2)(n_3 - m_3) + r.
\end{align*}
We show that the number of relations is given by the number of solutions to the following linear system over $\F_2$:
\begin{equation} \label{eq:system-dim-intro}
\begin{alignedat}{3}
H_1^T S &= H_2^T T &&= H_3^T U \\
H_1 T &= H_2 S &&= H_3^T V \\
H_1 U &= H_2^T V &&= H_3 S \\
H_1^T V &= H_2 U &&= H_3 T
\end{alignedat}
\end{equation}
where $H_1$ is shorthand for $H_1 \otimes \1 \otimes \1$ and similarly for $H_2, H_3$.
While the general case seems out of reach for the moment, we compute the dimension of the XYZ product code when the three parity-check matrices are invertible. 
\begin{theo}\label{thm:xyz-dim}
If the parity-check matrices $H_1$, $H_2$ and $H_3$ are invertible, then the dimension of the XYZ product code is equal to the number of solutions of the tensor Sylvester equation
\[
H_1 H_1^T X
= H_2 H_2^T X
= H_3 H_3^T X,
\]
which is equal to
\[
\sum_{i,j,k} \mathrm{deg}(\mathrm{gcd}(p^1_i,p^2_j,p^3_k)),
\]
where $p^\ell_i$ is the characteristic polynomial of the $i$-th Jordan block of $H_i H_i^T$.
\end{theo}
\noindent
This is established by proving a tensor version of the Cecioni-Frobenius theorem (Theorem \ref{thm:syl-3}) which originally counted the number of solutions to the (homogeneous) matrix Sylvester equation $AX = XB$.\\

Since one of our goals is to better understand the distance properties of XYZ product codes, it makes sense to first focus our attention on codes with dimension 1. We provide sufficient conditions for this (see Corollary \ref{cor:T}).
In particular, by the theorem above, it is sufficient that the three parity-check matrices $H_1, H_2, H_3$ are invertible and that the matrices $H_1 H_1^T, H_2 H_2^T, H_3 H_3^T$ have a unique common eigenvalue (each with geometric multiplicity 1).

\subsubsection{Minimum distance of the XYZ product codes}

For an XYZ product code with dimension 1 it is sufficient to consider each of the 3 logical operators ($X$, $Y$, $Z$-logical operators) to compute the minimum distance. 
We first give a $\Omega(N^{1/3})$ lower bound on the distance (Lemma \ref{lem:standard}) via a standard argument. 
Our main result is a proof that the minimum distance is equivalent to an elementary combinatorial problem involving binary 3-tensors.
More precisely, in the slightly simpler case where the parity-check matrices are sparse and symmetric (in addition to the conditions above), the minimum distance is given (up to a constant) by the following quantity (see Theorem \ref{thm:decoupling}):
\begin{align}\label{eqn:dec1}
\min_{(i,j,k) \in (1,2,3) } \min_M \big(  |(H_i^2 + H_j^2) M| + |(H_i^2 + H_k^2) M + R| \big)
\end{align}
where $R$ is the 3-tensor $R_{i,j,k} = \delta_{k,1}$, $|. |$ is the Hamming weight of the tensor, and we minimize over permutations of indices $(i,j,k)$ and binary 3-tensors $M \in \F_2^{n_1 \times n_2 \times n_3}$.
Equivalently, we may ask how closely we can approximate the \textit{heterogeneous} tensor Sylvester equation
\[
H_i^2 M
= H_j^2 M
= H_k^2 M + R
\]
for $(i,j,k) \in (1,2,3)$.
It is clear, by taking $M = 0$, that the minimum distance is at most $O(N^{2/3})$.
Intuitively, it seems unlikely that the minimum in Eqn.~\eqref{eqn:dec1} can be much lower than $N^{2/3}$. For the second term in \eqref{eqn:dec1} to be $o(N^{2/3})$, the tensor $M$ should be dense, that is $|M| = \Omega(N)$. This implies in turn that each slice $M_k$ (perpendicular to the $z$-direction) should be very structured in order to get $| (H_i^2 + H_j^2)M_k| =o(N^{1/3})$.
In that sense it seems likely that for a generic choice of $H_1, H_2, H_3$, the minimum of Eqn.~\eqref{eqn:dec1} is $\Theta(N^{2/3})$, but proving this claim has remained elusive so far. 

Expansion properties of the base codes are often helpful to prove nontrivial bounds on the distance, \textit{e.g.} for quantum expander codes \cite{LTZ15} or for lower bounding the cohomological distance of fiber bundle codes \cite{HHO20}, but seem unable to get us past $\Omega(N^{1/3})$ for the XYZ product codes. 
We also note that the probabilistic method which was successfully applied in \cite{FH14} and \cite{HHO20} for instance, is difficult to implement here, because of the algebraic nature of some of the XYZ product code properties.

\subsubsection{Cyclic XYZ product codes, and 3-dimensional embeddings}

A special class of XYZ product codes that can be analyzed in greater detail corresponds to taking classical cyclic codes, that is codes defined by circulant parity-check matrices. 
Such matrices can be conveniently described by a polynomial formalism since they are polynomials in the matrix
\[ \Omega_{n} = \left[ \begin{smallmatrix}
0 & 0 & \ldots & 0 & 1\\
1 & 0 & &  & 0\\
\vdots & \ddots &&& \vdots\\
0 & & & 1& 0
\end{smallmatrix}\right],\]
satisfying $(\Omega_n)^n=\1$. In other words, a circulant matrix is described by a polynomial over the finite ring $\F_2[x]/(x^n+1)$.
In this setup, the XYZ code becomes translation invariant. In addition, if the parity-check matrices are local in 1 dimension, then the resulting XYZ product code can be embedded on a three-dimensional lattice with local interactions. This is a simple generalization of the Chamon code also suggested in Ref.~\cite{BLT11}.

For these codes, computing the minimum distance \emph{seems} easier. 
A general 3-tensor corresponds to a polynomial 
$P(x,y,z) \in \F_2[x,y,z]/(x^{n_1}+1, y^{n_2}+1, z^{n_3}+1)$.
The interpretation of the polynomial is that the presence of some arbitrary monomial $x^i y^j z^k$ in $P$ means that there is a 1 in the cell indexed by $(i,j,k)$ in the binary tensor. For instance, the natural $Z$-logical operator, corresponding to a full horizontal plane of $\sigma_3$ operators, is described by the polynomial 
$R(x,y,z) = \sum_{i=0}^{n_1-1} \sum_{j=0}^{n_2-1} x^i y^j$.
In that case, the minimum distance is given, up to constant factors, by an optimization problem of the form:
\[ \min_P \big(  |(P_1(x) + P_2(y)) P(x,y,z)| + |(P_1(x) + P_3(z)) P(x,y,z) + R(x,y)| \big).\]
In particular, the general form of polynomials $P(x,y)$ such that $|(P_1(x) +P_2(y)) P(x,y)|$ is small is rather well understood. It is a sum of \emph{fractal operators} $(P_1 + P_2)^{2^p -1}$ since
\[ (P_1(x) + P_2(y)) (P_1(x) + P_2(y))^{2^p -1}= (P_1(x) + P_2(y))^{2^p} = P_1(x^{2^p}) + P_2(y^{2p}))\]
where we used that $P(x)^2 = P(x^2)$ over $\F_2[x]$. 
Since the code is LDPC, we get that
\[
\left|(P_1(x) + P_2(y)) (P_1(x) + P_2(y))^{2^p -1}\right| = O(1).
\]
We note that the existence of such fractal operators (which are also central in the study of Haah's cubic code \cite{haa11} for instance) immediately show that the codes are not locally testable since there exist large weight errors with constant weight syndromes.

In Section \ref{sec:simplest}, we consider a specific instance of the XYZ product code associated to parity-check matrices  $H_i = \1_{n_i} + \Omega_{n_i} + \Omega_{n_i}^T$. If the $n_i$'s are odd and no multiple of 3, we show that the code dimension is 
\[ 4(\gcd(n_1,n_2,n_3)-1) + 1.\]
Choosing the $n_i$'s to be coprime yields a single logical qubit and the optimization problem takes the simple form 
\[ \min_P | (1+xy)(1+x/y)P | + |(1+xz)(1+x/z)P + R|,\]
which can be analyzed in some detail. Unfortunately, it seems that many choices of $n_i$'s lead to a distance $O(N^{1/3})$ suggesting that higher weight polynomials (weight 5) may be required to reach much larger distances.

The existence of fractal operators also hints at the possibility that some of these codes might display a large (logarithmic) energy barrier, and therefore be potentially useful as self-correcting quantum memories \cite{BLP16, ter15}. While we do not have a general answer to these questions at the moment, we prove in Section \ref{sec:simplest} that 3D codes obtained by taking identical circulant parity-check matrices have a constant energy barrier, even if the $n_i$'s are coprime. Such codes are therefore unable to passively store quantum information. 

\subsection{Open questions}

Our work leads to a large number of open questions, and we give a few below:

\begin{itemize}
\item The most obvious question is whether specific instances of XYZ product codes do achieve a minimum distance of $\Theta(N^{2/3})$. If this is the case, is it also possible when restricting to codes that are embeddable in 3 dimensions? If so, this would saturate an upper bound due to Bravyi and Terhal \cite{BT09}. 
It would also be interesting to understand how the distance behaves when the code dimension increases. For instance, it is easy to get quantum codes with constant rate (linear code dimension) by taking full rank rectangle parity-check matrices of size $m \times n$ with $m/n$ bounded away from 0 or 1, but it is even less clear what the distance is in that case. 

\item It would be interesting to also study the lifted product variant of the XYZ product codes, following the work of \cite{PK20}. This could in principle yield codes of distance beyond $N^{2/3}$.

\item It would also be nice to understand whether there exist triples of parity-check matrices such that the XYZ product code $\mathcal{Q}(H_1, H_2, H_3)$ is not a valid stabilizer code, because $-\1$ belongs to the stabilizer group.

\item Another open question is to determine the dimension of the XYZ product code in the general case. In particular, the case where $H_i H_i^T$ is not full rank appears harder to analyze. We comment on the difficulties of getting a full solution to this problem right before Section \ref{sub:sylvester}.

\item The Chamon code initially appeared in literature on quantum many-body systems. Similarly, the models corresponding to cyclic XYZ product codes embeddable in 3 dimensions may display interesting properties. For instance, the appearance of fractal operators is reminiscent of Haah's cubic code and raises the question of whether some XYZ product codes display some form of self-correction, or a logarithmic energy barrier scaling? A first step in this direction would be to show that some variants do not have any string logical operators.

\item Replacing qubits by qudits, or considering XYZ-type products in higher dimensions also appears as a potentially fruitful avenue of research. For instance, a, XYZ-type product in $2D+1$ dimensions can be defined in exactly the same fashion as soon as we have $2D+1$ mutually anticommuting operators, which is possible by replacing each qubit by $D$ qubits. 

\item A question that we did not explore at all is that of decoding XYZ product codes. In particular, it is tempting to try to adapt decoders that work well for hypergraph product codes, for instance the small-set-flip decoder \cite{LTZ15}. Maybe better decoders are also possible for the topological XYZ product codes embedded in 3 dimensions?

\end{itemize}

\subsection*{Organization}  

Section \ref{sec:prelim} introduces in detail the hypergraph product code construction which is essential for the XYZ product code construction that we present in Section \ref{sec:construction}. We also discuss a CSS version of the construction and compare it with the other 3-fold hypergraph product code constructions that have already appeared in the literature. 
Section \ref{sec:dim} is devoted to the study of the dimension of the XYZ product code and Section \ref{sec:decoupling} to its minimum distance. We consider cyclic, translation-invariant instances of the code in Section \ref{sec:cyclic} and study in more detail the simplest such instance in Section \ref{sec:simplest}.

\subsection*{Acknowledgments}
We would like to thank Benjamin Audoux, Alain Couvreur, Denise Maurice, Barbara Terhal, Jean-Pierre Tillich and Gilles Z\'emor for many fruitful discussions on quantum LDPC codes. Most of this work was done while SA was in the CWI-Inria lab.
AL and CV acknowledge support from the ANR through the QuantERA project QCDA and from the Plan France 2030 through the project ANR-22-PETQ-0006.

\newpage
\section{Preliminaries and notations}
\label{sec:prelim}

\subsection{Notations}

\begin{itemize}
\item The Pauli group on one qubit is $\P = \langle X, Y, Z\rangle$
where $X, Y=i XZ, Z$ are the Pauli matrices given by
\[ X = \begin{pmatrix} 0 & 1\\ 1 & 0\end{pmatrix}, \quad  Y =
\begin{pmatrix} 0 & -i\\ i & 0\end{pmatrix}, \quad  Z = \begin{pmatrix}
1 & 0\\ 0 & -1\end{pmatrix}. \]
\item We will overload the variables $X, Y, Z$ and sometimes write
$\sigma_1, \sigma_2, \sigma_3$ instead of $X, Y, Z$ to avoid any
confusion or when it is convenient.
\item The multi-qubit Pauli group on $n$ qubits, $\mathcal{P}_n$,
consists in tensor product of single-Pauli operators.
\item We will denote the $n\times n$ identity matrix as $\1_n$ and often
use just $\1$ where the dimension should be clear from context.
\item We will manipulate tensors of binary numbers as well as tensors of
Pauli operators.
We use standard capital roman letters to denote binary tensors like,
$A$, $B$, $C$, $D$, $S$, $T$, $U$, $V$, $M$, and calligraphic letters
for Pauli tensors like $\mathcal{A}$, $\mathcal{B}$, $\mathcal{C}$,
$\mathcal{D}$, $\mathcal{S}$, $\mathcal{T}$, $\mathcal{U}$, $\mathcal{V}$.
For instance given a binary tensor $A$, we can define a Pauli tensor of
the same dimension $\mathcal{A}$ as
\begin{equation}
    \mathcal{A} = \sigma_1^A,\label{eq:tensornotation}
\end{equation}
where this notation means that for any entry where $A$ is $1$,
respectively $0$, $\mathcal{A}$ has a Pauli $\sigma_1$, respectively
$\1$, at the corresponding entry.

\item The Hamming weight of a binary tensor $A$ is the number of nonzero
entries in $A$ and is denoted $\left \vert A\right \vert$.
\item The Pauli weight of a Pauli tensor $\mathcal{A}$ is the number of
entries that are not the identity and is also denoted as $\left
\vert\mathcal{A}\right \vert$.
In the case of \eqref{eq:tensornotation} for instance the Hamming weight
of $A$ and the Pauli weight of $\mathcal{A}$ coincide.

\item The commutator of two Pauli tensors of the same dimensions,
$\mathcal{A}$ and $\mathcal{B}$, is $-\1$ if they anti-commute and $+\1$
if they commute and is denoted as $[\mathcal{A},\mathcal{B}]$.
Computing it can be done using the following
\begin{equation}
    [\mathcal{A},\mathcal{B}]
=\mathcal{A}\cdot\mathcal{B}\cdot\mathcal{A}^{-1}\cdot\mathcal{B}^{-1},
\end{equation}
where the product is taken element wise and the right hand side is
therefore always either $+\1$ or $-\1$.

\end{itemize}

\subsection{Hypergraph Product code} \label{sec:hypergraph-codes}

As mentioned in the introduction, CSS codes are naturally expressed as
chain complexes of length 3 or subcomplex of length 3 taken from longer
chain complexes.
This point of view hints at constructing codes out of two chain
complexes, as short as length 2, using the homological product of chain
complexes.
Length 2 chain complexes then just correspond to classical codes.
This was first proposed by Tillich and Zemor in \cite{TZ13} under the
name hypergraph product codes.
Generalizations and further studies of this construction have been made
to consider longer initial chain complexes in \cite{BH14, AC19}.

More specifically, we can derive the hypergraph product code from a pair
of classical parity-check matrices $H_1$ and $H_2$.
If we interpret these binary matrices as length 2 chain complexes, then
the hypergraph product code is described
as depicted in Figure \ref{fig:chaincomplex-hpc}.
\begin{figure}[!h]
\centering
\includegraphics[width=0.6\linewidth]{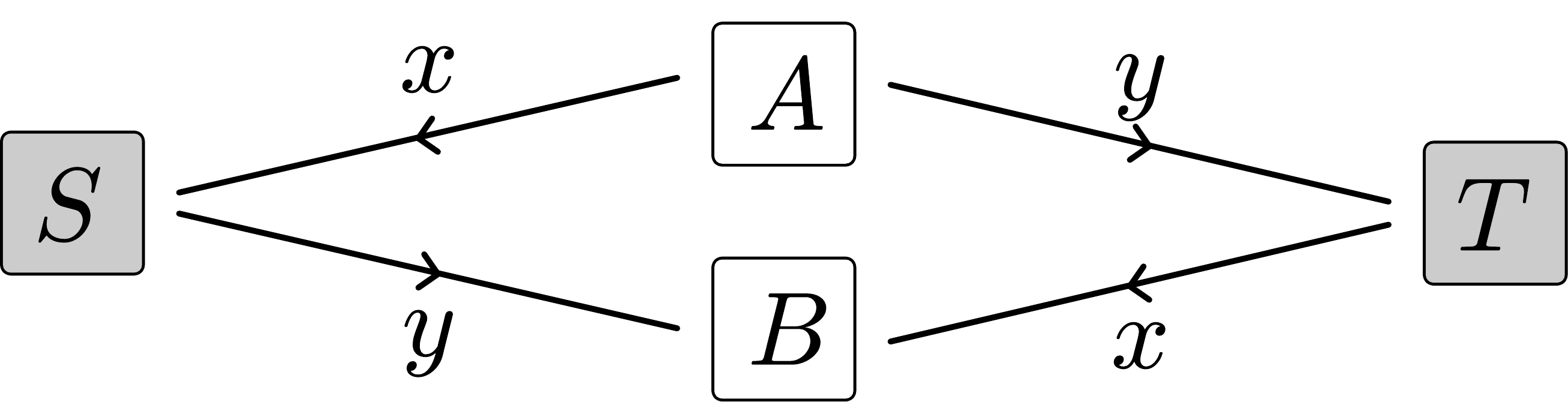}
\caption{``Chain complex'' representation of the hypergraph product
code. The $x$ label corresponds to the operator $H_1 \otimes \1$ and the
$y$ label corresponds to the operator $\1 \otimes H_2$.} 
\label{fig:chaincomplex-hpc}
\end{figure}

We identify the resulting vector spaces with matrices or 2-tensors $A$,
$B$, $S$ and $T$.
If $H_i$ has dimension $m_i \times n_i$, for $i=1,2$, then these tensors
have dimensions
\[
A \in \F_2^{n_1 \times n_2}, \quad
B \in \F_2^{m_2 \times m_3}, \quad
S \in \F_2^{m_1 \times n_2}, \quad
T \in \F_2^{n_1 \times m_2}.
\]
The boundary operators indexed $x$ and $y$ correspond to operators $H_1
\otimes \1$ and $\1 \otimes H_2$.
The boundary operators of the resulting chain complex are hence
$\partial = [H_1 \otimes \1, \1 \otimes H_2]^T$ and $\partial' = [\1
\otimes H_2, H_1 \otimes \1]$, and we have the boundary condition
$\partial' \, \partial = 2 H_1 \otimes H_2 = 0$.

We can derive a quantum code from this chain complex by associating
qubits to the matrix elements of $A$ and $B$, and checks to the matrix
elements of $S$ and $T$.
We will describe a general check on the qubits in $A$ and $B$ by a
tensor over the Pauli group
\[
\begin{bmatrix} {\cal A} \\ {\cal B} \end{bmatrix}
\in \begin{bmatrix} \P^{n_1 \times n_2} \\ \P^{m_1 \times m_2}
\end{bmatrix}.
\]
The checks are generated by the binary tensors $S$ and $T$.
Usually the code is chosen to be of the CSS type by choosing one between
$S$ and $T$ to represent $X$-type checks and the other $Z$-type checks.
Here we are going to consider a Clifford equivalent code which will
generalize more easily to the XYZ product codes with checks of mixed $X$
and $Y$ types.
The checks generated by $S$ correspond to $\sigma_1$ checks on the
qubits in $A$ with pattern $(H_1^T \otimes \1) S$, and $\sigma_2$ checks
on the qubits in $B$ with pattern $(\1 \otimes H_2) S$.
Similarly, the checks from $T$ correspond to $\sigma_1$ checks on $B$
with pattern $(H_1 \otimes \1) T$ and $\sigma_2$ checks on $A$ with
pattern $(\1 \otimes H_2^T) T$.
For a general check $(S,T)$, we use the shorthand
\[
\begin{bmatrix} {\cal A} \\ {\cal B} \end{bmatrix}
= \begin{bmatrix}
\sigma_1^{(H_1^T \otimes \1) S} \; \sigma_2^{(\1 \otimes H_2^T) T} \\
\sigma_2^{(\1 \otimes H_2) S} \; \sigma_1^{(H_1 \otimes \1) T}
\end{bmatrix}.
\]
We will often write this more abstractly as
\[
\begin{bmatrix} {\cal A} \\ {\cal B} \end{bmatrix}
=
\begin{bmatrix}
\uX^\dag & \uY^\dag \\
\uY & \uX
\end{bmatrix}
\begin{bmatrix}
S \\ T
\end{bmatrix}
\]
where we introduce the operators $\uX = \sigma_1^{H_1 \otimes \1}$ and
$\uY = \sigma_2^{\1 \otimes H_2}$.
We let a dagger $\uX^\dag$ denote $\uX^\dag = \sigma_1^{H_1^T
\otimes \1}$, a product $\uX S$ should be interpreted as $\sigma_1^{(H_1
\otimes \1) S}$ and a sum $\uY S + \uX T$ corresponds to a product
$\sigma_2^{(\1 \otimes H_2) S} \; \sigma_1^{(H_1 \otimes \1) T}$.

Pairs of checks indexed by tensors of the type $(S,0)$ and $(S^\prime,0)$ (or $(0,T)$ and $(0,T^\prime)$) always commute as they are the same Pauli type on the same blocks. One then simply needs to check that any pair of checks, indexed by
tensors $(S,0)$ and $(0,T)$, also commutes. 
Note first that for any 2-tensors $A$ and $B$, it holds that 
\[ [\sigma_1^A, \sigma_2^B] = (-1)^{\sum_{i,j} A_{ij} B_{ij}} \1 = (-1)^{ \mathrm{tr} (A^TB)} \1,\]
since the corresponding operators anticommute iff the supports of $A$ and $B$ have an odd overlap.
The commutator of the tensors indexed by $(S,0)$ and $(0,T)$, given by
\[
\begin{bmatrix} {\cal A}_1 \\ {\cal B}_1 \end{bmatrix}
=
\begin{bmatrix}
    \sigma_1^{(H_1^T \otimes \1) S}\\
    \sigma_2^{(\1 \otimes H_2) S}
\end{bmatrix},\qquad \begin{bmatrix} {\cal A}_2 \\ {\cal B}_2 \end{bmatrix}
=
\begin{bmatrix}
    \sigma_2^{(\1 \otimes H_2^T) T} \\
     \sigma_1^{(H_1 \otimes \1) T}
\end{bmatrix},
\]
can be computed as follows:
\begin{align*}
    \left [\begin{bmatrix} {\cal A}_1 \\ {\cal B}_1 \end{bmatrix},
\begin{bmatrix} {\cal A}_2 \\ {\cal B}_2 \end{bmatrix}\right ] &=
(-1)^{\mathrm{tr} (S^T (H_1 \otimes \1) (\1 \otimes H_2^T) T + S^T (\1 \otimes H_2^T)(H_1 \otimes \1) T)}\1 = \1.
\end{align*}

Remark that if the two classical codes $H_1$ and $H_2$ are repetition
codes then one gets the so called $XZZX$ toric code (or more precisely
$XYYX$ here) that has been recently highlighted for its good performance
against all types of Pauli noise \cite{BTB20}.


\section{The XYZ product code construction}
\label{sec:construction}

\subsection{XYZ product code}
\label{sub:homological}

The XYZ product code construction is a generalization of the hypergraph product code involving a third classical code which will enforce Pauli $Y$-type constraints, thus making the quantum code \emph{not CSS}.
Given three parity-check matrices $H_1$, $H_2$ and $H_3$, we again start by defining a chain complex, depicted in Figure \ref{fig:chaincomplex}.

\begin{figure}[!h]
\centering
\includegraphics[width=0.8\linewidth]{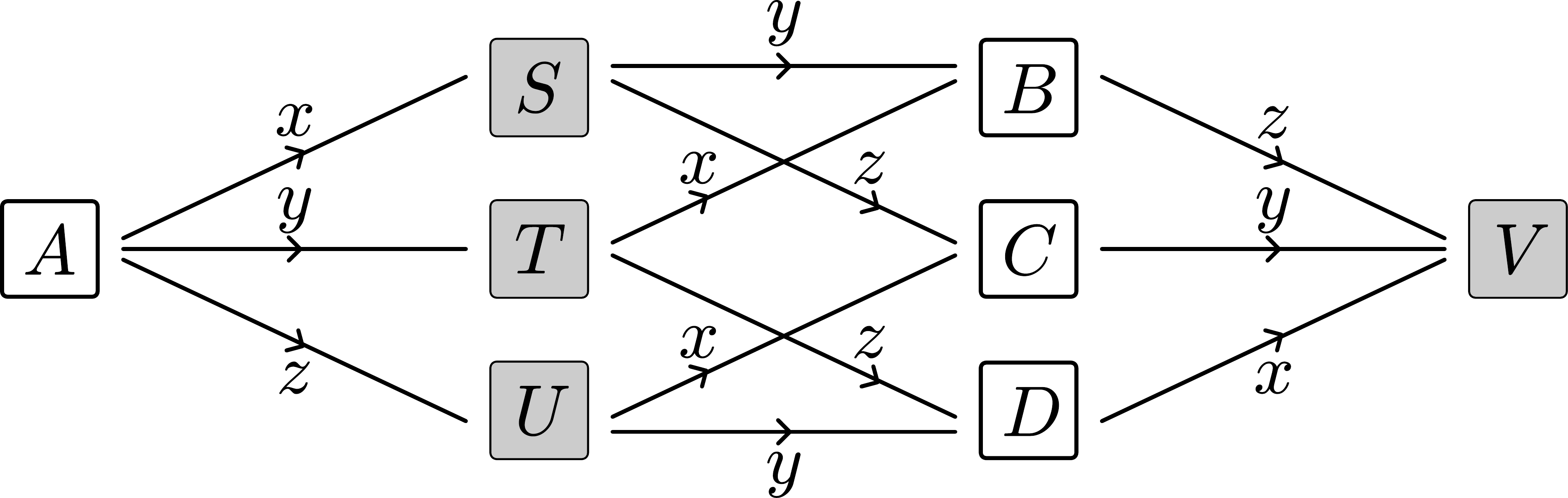}
\caption{``Chain complex'' representation of the XYZ product codes. The $x$ label corresponds to the operator $H_1 \otimes \1 \otimes \1$, the $y$ label to the operator $\1 \otimes H_2 \otimes \1$ and the $z$ label to the operator $\1 \otimes \1 \otimes H_3$.}  
\label{fig:chaincomplex}
\end{figure}

We proceed similarly to the hypergraph product code construction.
We identify the resulting vector spaces with 3-tensors
\[
A \in \F_2^{n_1 \times n_2 \times n_3}, \quad
B \in \F_2^{n_1 \times m_2 \times m_3}, \quad
C \in \F_2^{m_1 \times n_2 \times m_3}, \quad
D \in \F_2^{m_1 \times m_2 \times n_3},
\]
which will index the qubits, and 3-tensors
\[
S \in \F_2^{m_1 \times n_2 \times n_3}, \quad
T \in \F_2^{n_1 \times m_2 \times n_3}, \quad
U \in \F_2^{n_1 \times n_2 \times m_3}, \quad
V \in \F_2^{m_1 \times m_2 \times m_3},
\]
which will index the checks.
The total number of qubits is
\begin{align}
\label{eqn:length}
N = \underbrace{n_1 n_2 n_3}_{n_A} +\underbrace{ n_1 m_2 m_3}_{n_B} +\underbrace{ m_1 n_2 m_3}_{n_C} +\underbrace{ m_1 m_2 n_3}_{n_D}.
\end{align}
A general stabilizer element is a tensor over the Pauli group written abstractly as
\begin{align}\label{eqn:gen}
\begin{bmatrix}
{\cal A} \\ {\cal B} \\ {\cal C} \\ {\cal D} \end{bmatrix} = 
\begin{bmatrix}
\uX^\dagger & \uY^\dagger & \uZ^\dagger & \cdot\\
\uY & \uX &\cdot& \uZ^\dagger\\
\uZ & \cdot& \uX & \uY^\dagger\\
\cdot& \uZ & \uY  &\uX^\dagger \end{bmatrix}
\begin{bmatrix} S\\T \\U \\V\end{bmatrix}
=: \Gamma(S,T,U,V).
\end{align}
Here we introduced the operators\footnote{We repeat that the dimension of the identity operator $\1$ will vary depending on the context.}
\[
\uX = \sigma_1^{H_1 \otimes \1 \otimes \1}, \qquad
\uY = \sigma_2^{\1 \otimes H_2 \otimes \1}, \qquad
\uZ = \sigma_3^{\1 \otimes \1 \otimes H_3}.
\]
We again use the convention that an adjoint $\uX^\dag$ denotes the operator $\uX^\dag = \sigma_1^{H_1^T \otimes \1 \otimes \1}$, a product $\uY S$ denotes $\sigma_2^{(\1 \otimes H_2 \otimes \1) S}$, and a sum $\uY S + \uX T$ denotes the product $\sigma_2^{(\1 \otimes H_2 \otimes \1) S} \sigma_1^{(H_1 \otimes \1 \otimes \1) T}$.
We call the resulting stabilizer code the \textit{XYZ product code}, and we refer to it by the shorthand ${\cal Q}(H_1,H_2,H_3)$.
In Figure \ref{fig:cubecube} we give a slightly alternative representation of the code or chain complex.
It better reflects how we will be able to embed the code into 3D for certain choices of $H_i$'s.

\begin{figure}[!h]
\centering
\includegraphics[width=0.6\linewidth]{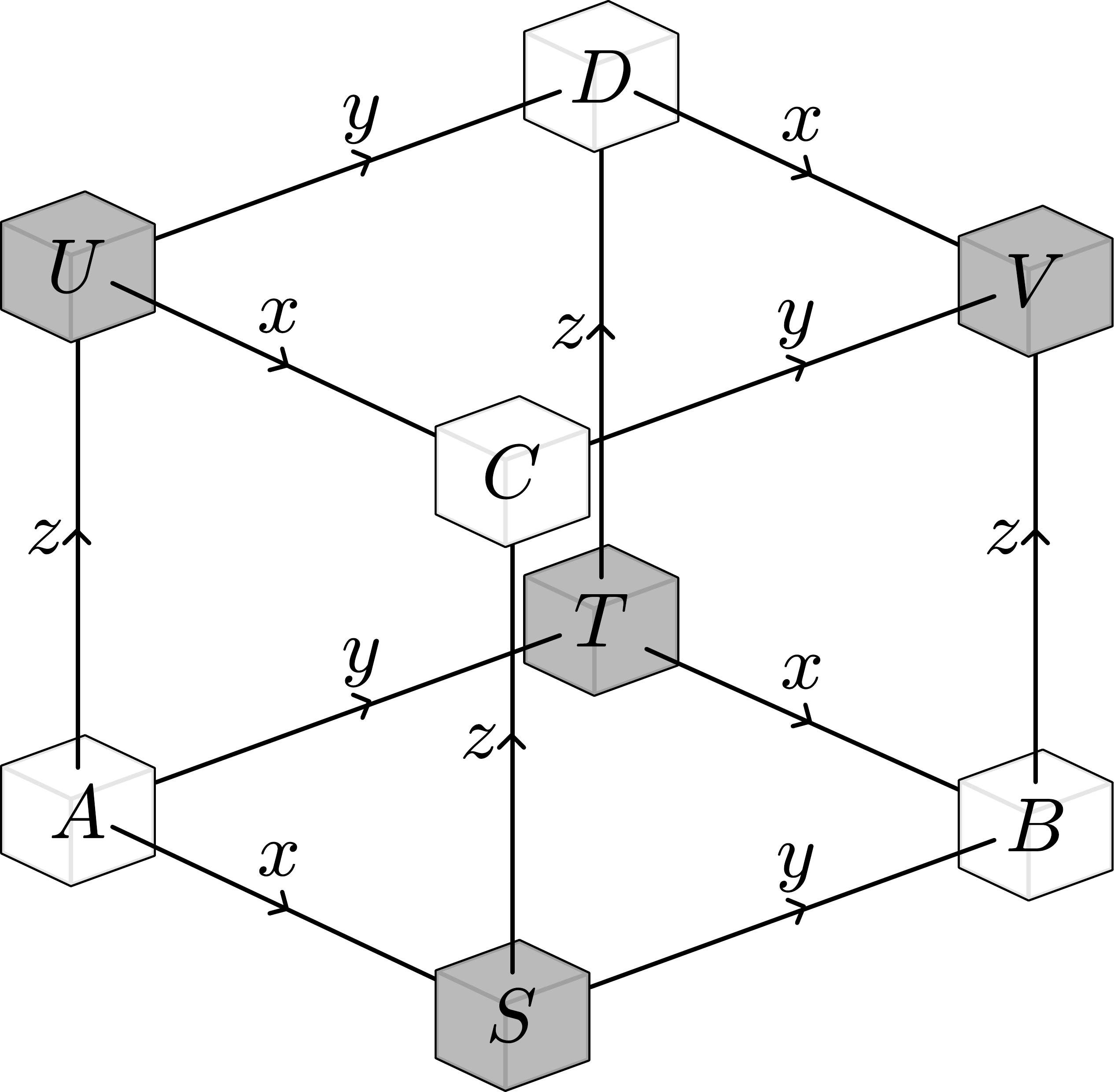}
\caption{3D structure of the XYZ product code: qubits indexed by the 3-tensors $A, B, C, D$, generators by $S, T, U, V$.}  
\label{fig:cubecube}
\end{figure}

Let us establish that the resulting stabilizer group is Abelian.
\begin{lemma}
\label{lem:xyz}
The stabilizer group of the XYZ product code is Abelian.
\end{lemma}

As noted before, one also needs to prove that $-\1$ is not in the stabilizer group to show that the code is readily a well defined stabilizer code. This question turns out to be rather subtle and we will discuss two possible solutions in the following: fixing the $\pm1$ phases of the generators, or considering the CSS version of the XYZ product code (see Sections \ref{sub:minusone} and \ref{sub:CSS} for details).

\begin{proof}
With $\Gamma$ as in \eqref{eqn:gen}, consider two arbitrary elements of the stabilizer
\[
({\cal A}_1, {\cal B}_1, {\cal C}_1, {\cal D}_1)^T = \Gamma (S_1,T_1, U_1, V_1)
\quad \text{ and } \quad
({\cal A}_2, {\cal B}_2, {\cal C}_2, {\cal D}_2)^T = \Gamma (S_2,T_2, U_2, V_2).
\]
We can compute their commutator:
\begin{align*}
[ ({\cal A}_1, {\cal B}_1, {\cal C}_1, {\cal D}_1)^T, ({\cal A}_2, {\cal B}_2, {\cal C}_2, {\cal D}_2)^T]
&=  [ ({\cal A}_1, {\cal B}_1)^T, ({\cal A}_2, {\cal B}_2)^T]
 + [ ({\cal A}_1,{\cal C}_1)^T, ({\cal A}_2, {\cal C}_2)^T] \\
& \quad +[ ({\cal A}_1, {\cal D}_1)^T, ({\cal A}_2, {\cal D}_2)^T]
 +[ ({\cal B}_1, {\cal C}_1)^T, ( {\cal B}_2, {\cal C}_2)^T] \\
& \quad +[ ({\cal B}_1, {\cal D}_1)^T, ( {\cal B}_2, {\cal D}_2)^T]
 +[ ({\cal C}_1, {\cal D}_1)^T, ({\cal C}_2, {\cal D}_2)^T].
\end{align*}
Now each of these small commutators corresponds to a commutator between $X$-type and $Z$-type generators in the standard hypergraph product code construction (see Section \ref{sec:hypergraph-codes}).
Consequently they all vanish. 
\end{proof}

We can alternatively use the notation $({\cal A}, {\cal B}, {\cal C}, {\cal D}) \in \P^{n_A+n_B+n_C+n_D}$ to denote a Pauli error.
Its syndrome can then be abstractly written as
\[
\begin{bmatrix}
S\\T \\U \\V \end{bmatrix} = 
\begin{bmatrix}
\uX & \uY^\dagger & \uZ^\dagger & \cdot\\
\uY & \uX^\dagger &\cdot& \uZ^\dagger\\
\uZ & \cdot& \uX^\dagger & \uY^\dagger\\
\cdot& \uZ & \uY& \uX
\end{bmatrix}
*
\begin{bmatrix}
{\cal A} \\ {\cal B} \\ {\cal C}\\ {\cal D}
\end{bmatrix}
=: \Sigma * \begin{bmatrix}
{\cal A} \\ {\cal B} \\ {\cal C}\\ {\cal D}
\end{bmatrix}.
\]
We call $\Sigma$ the ``parity-check'' operator of the XYZ product code. 
The star product $*$ between a pair of Pauli operators is
\[
\sigma_i *\sigma_j
= \delta_{i\ne j},
\]
and we extend it to a pair of tensors of equal size (say $\uX * {\cal A}$) by taking the star product elementwise.
The product gives 0 where two Pauli operators commute, and it gives 1 where they commute.
This implies that $\uX$ detects only errors that do not commute with $X$ (and similarly for $\uY$, $\uZ$).

Now we can check that the syndrome of a stabilizer vanishes by computing the operator $\Sigma * \Gamma $:
\[\Sigma * \Gamma= \left[\begin{smallmatrix}
\uX * \uX^\dagger + \uY^\dagger  * \uY+ \uZ^\dagger * \uZ & \uX* \uY^\dag + \uY^\dag* \uX & \uX *\uZ^\dag + \uZ^\dag *\uX & \uY^\dag * \uZ^\dag +\uZ^\dag *\uY^\dag \\
\uX^\dagger * \uY + \uY*\uX^\dagger& \uX^\dagger* \uX + \uY^\dagger *\uY + \uZ *\uZ^\dagger & \uY* \uZ^\dagger+\uZ^\dagger*\uY & \uX^\dagger* \uZ^\dagger+\uZ^\dagger *\uX^\dagger \\
\uX^\dagger* \uZ+\uZ *\uX^\dagger & \uY^\dagger* \uZ+ \uZ *\uY^\dagger & \uX^\dagger*\uX + \uY^\dagger* \uY + \uZ* \uZ^\dagger &  \uX^\dagger* \uY^\dagger+\uY^\dagger * \uX^\dagger \\
 \uY* \uZ+\uZ *\uY & \uX* \uZ+\uZ*\uX & \uX * \uY+\uY*\uX, & \uX* \uX^\dagger + \uY *\uY^\dagger + \uZ *\uZ^\dagger 
\end{smallmatrix}\right]=0,
\]
where we used that \textit{e.g.}~$\uX * \uX^T = 0$ (because the $\sigma_1$-operators commute) and that $\uX * \uY + \uY * \uX = 2 \uX * \uY = 0$ (recall that the dimensions of $\uX$ depend on its position in the matrix).


\subsection{The case of $-\1\in\mathcal{S}$}
\label{sub:minusone}
The usual prescription to define a stabilizer code is to pick an Abelian subgroup of the $n$-qubit Pauli group $\mathcal{S}\subset\mathcal{P}_n$ which does not contain $-\1$.
With such a choice it is guaranteed that there exists a common $+1$ eigenspace to all the elements of $\mathcal{S}$ which is defined to be the code space.
As mentioned above, for the XYZ product code construction, we are only able to prove that $-\1\not\in\mathcal{S}$ in some restricted cases.
We show however, that $-\1\in \mathcal{S}$ is not in fact a problem to define a code.

If $\mathcal{S}$ contains $-\1$ then there cannot be a common $+1$-eigenspace to all the elements in $\mathcal{S}$.
Nevertheless, the generators of $\mathcal{S}$ are Hermitian operators, \textit{i.e.}, observables of the system, and the fact that $\mathcal{S}$ is Abelian ensures that they are all simultaneously diagonalizable and can therefore be measured simultaneously.
That $-\1$ is in $\mathcal{S}$ simply prevents the measurement outcomes to all be $+1$.
To still use such a group to define a code one just has to settle on a consistent convention for which elements of $\mathcal{S}$ should yield $+1$ and which ones $-1$ when measured.

To explain in more details how to make this choice consider an Abelian subgroup $\mathcal{S}\subset\mathcal{P}_n$ which contains $-\1$.
We can consider $\mathcal{S}$ up to phases by looking at the quotient subgroup $\mathcal{S}^\pm = \mathcal{S}/\langle -\1\rangle$, since $\langle -\1\rangle$ is a normal subgroup of $\mathcal{S}$.
We can first pick an independent generating set for $\mathcal{S}^\pm$:
\begin{equation}
	\left \langle g^\pm_1,\ldots,g^\pm_{n-k}\right \rangle = \mathcal{S}^\pm.\label{eq:genSpm}
\end{equation}
Then for each generator $g^\pm_i\in\mathcal{S}^\pm$ choose any corresponding representative $g_i\in\mathcal{S}$ (there are two choices $g_i$ and $-g_i$).
They give directly an independent generating set for $\mathcal{S}$: 
\[
	\mathcal{S} = \left \langle-\1,g_1,\ldots,g_{n-k}\right \rangle.
\]
The subgroup $\tilde{\mathcal{S}} = \left \langle g_1,\ldots,g_{n-k}\right \rangle$ is a valid stabilizer subgroup of $\mathcal{S}$ since it does not contain $-\1$ (by the independence assumption on the $g_i$'s) and we can finally define the code to be the stabilizer code stabilized by $\tilde{\mathcal{S}}$.
If the initial $\mathcal{S}$ was given by an original generating set
\[
	\mathcal{S} = \left \langle g_1^{\rm orig.},\ldots,g_m^{\rm orig.}\right \rangle,
\]
for some $m$, then for each of them, either $+g_i^{\rm orig.}\in\tilde{\mathcal{S}}$ or $-g_i^{\rm orig.}\in\tilde{\mathcal{S}}$ and checking which is the case allows to keep the original checks and sets the convention to be adopted for syndrome measurements and error correction.
The dimension of the code is given by the number of independent generators of $\mathcal{S}^\pm$ and is independent of the choice of signs for $\tilde{\mathcal{S}}$, provided that they are consistent.
Similarly the minimal distance is independent of the choice of signs, since commutation relations are not modified by signs, so that the parameters $\llbracket n,k,d \rrbracket$ are always well defined even if $-\1\in\mathcal{S}$.

A practical way to obtain a basis for $\mathcal{S}^\pm$ in \eqref{eq:genSpm} is to use the mapping to CSS codes described below in Section \ref{sub:CSS}, Eq.~\eqref{eq:emulPauli}.
This mapping forgets phases and allows one to simply find a basis for $\mathcal{S}^\pm$ using standard binary linear algebra.

\subsection{A CSS version of the construction}
\label{sub:CSS}

The XYZ product code construction yields a stabilizer code which is not CSS. In particular each generator involves all three Pauli matrices. Besides there is no generic transversal Clifford operation that can transform it into a CSS code since each qubit is \textit{a priori} acted on by different stabilizer generators with the three different types of Pauli.
However, the potentially good asymptotic parameters of the XYZ product code are not tied to this non-CSS nature, and we can in fact derive a CSS code with the same parameters up to constant factors.
A recipe to turn any $\llbracket n,k,d\rrbracket $ stabilizer code into a $\llbracket 4n, 2k, d\rrbracket $ CSS stabilizer code was devised in \cite{BTL10} using a Majorana fermion code as an intermediary code.
The transformation also preserves locality and LDPC properties.
We recall it here directly, skipping the intermediary Majorana fermion code.

The gist of the idea is to emulate the one-qubit Pauli group (ignoring phases) with just one type of Pauli operators acting on four qubits using the following mapping that we denote $\phi_i$ for $i\in\{1,2,3\}$ for each choice of Pauli:
\begin{equation}
\begin{tabular}{l c c c}
	$\phi_i:$& $\mathcal{P}$ &$\longrightarrow$& $\mathcal{P}^{\otimes 4}$\\
	            & $\1$ &$\mapsto$ & $\1\otimes\1\otimes\1\otimes\1$\\
	            & $\sigma_1$ &$\mapsto$ & $\sigma_i\otimes\sigma_i\otimes\1\otimes\1$\\
	            & $\sigma_2$ &$\mapsto$ & $\sigma_i\otimes\1\otimes\sigma_i\otimes\1$\\
	            & $\sigma_3$ &$\mapsto$ & $\sigma_i\otimes\1\otimes\1\otimes\sigma_i$
\end{tabular}.\label{eq:emulPauli}
\end{equation}
Note that the product of the three mapped Pauli operators is the four-body Pauli $\sigma_i\otimes\sigma_i\otimes\sigma_i\otimes\sigma_i$.
This indicates that this operator should be enforced as a stabilizer.
With this stabilizer taken into account, the mapping correctly reproduces the multiplicative law of the group but not the commutation relations which become trivial.
The commutation relation between two mapped elements with a different choice of Pauli for the mapping are correctly reproduced
\begin{align}
	\forall i,j,k,\ell\in\{1,2,3\}\quad\phi_i(\sigma_j)\phi_k(\sigma_\ell) &= (-1)^{\delta_{i\ne k}\delta_{j\ne \ell}}\phi_k(\sigma_\ell)\phi_i(\sigma_j)\label{eq:emulcommutation}\\
	\phi_i(\1)\phi_k(\sigma_\ell) &= \phi_k(\sigma_\ell)\phi_i(\1).\label{eq:emulcommutationid}	
\end{align}
The mapping can be straightforwardly extended to the multi-qubit Pauli group mapping independently each single-qubit Pauli which we also denote $\phi_i:\mathcal{P}^{\otimes n}\longrightarrow\mathcal{P}^{\otimes 4n}$.

We can now describe the transformation from a $\llbracket n,k,d\rrbracket $ stabilizer code $\mathcal{C}$ to a $\llbracket n_\text{CSS}=4n,k_\text{CSS}=2k,d_\text{CSS}=2d \rrbracket $ CSS code $\mathcal{C}_{\text{CSS}}$ in three steps:
\begin{enumerate}
	\item For each qubit of the original code, create four qubits.
	\item Enforce $\sigma_1\otimes\sigma_1\otimes\sigma_1\otimes\sigma_1$ and $\sigma_3\otimes\sigma_3\otimes\sigma_3\otimes\sigma_3$ as stabilizers on each group of four qubits created.\label{step:4body}
	\item For each stabilizer generator of the original code, say $S\in\mathcal{P}^{\otimes n}$, enforce $\phi_1(S)$ and $\phi_3(S)$ as stabilizers.\label{step:stab}
\end{enumerate}

The number of qubits in the new code is indeed $n_\text{CSS}=4n$.
It is a valid stabilizer code per Eqn.~\eqref{eq:emulcommutation} and \eqref{eq:emulcommutationid} and it is CSS since all the stabilizers created with $\phi_1$ are fully $X$-type and the ones created with $\phi_3$ fully $Z$-type.
To compute the dimension of the new code we have to count the number of independent stabilizers.
The four body stabilizers added at step \ref{step:4body} form $2n$ independent stabilizers.
With these stabilizers taken into account, the mapping then correctly reproduces the multiplicative law of the Pauli group, hence step \ref{step:stab} creates twice as many independent stabilizers as in the original code, that is $2(n-k)$.
This makes for $k_\text{CSS} = 4n - 2n -2(n-k) = 2k$ logical qubits.
Finally for the distance, any logical Pauli operator of $\mathcal{C}_\text{CSS}$ has to commute with the four-body checks of step \ref{step:4body} and hence act on an even number of qubits from each group of four representing one original qubit.
Hence it can always be mapped to a logical operator of the original code which has to have weight at least $d$ and we conclude that $d_\text{CSS} = 2d$.

This establishes the following:
\begin{lemma}[\cite{BTL10}]\label{lem:nonCSS2CSS}
A $\llbracket n,k,d\rrbracket$ stablizer code can be mapped to a $\llbracket 4n, 2k, d\rrbracket$ CSS code with generator weights not larger than twice the weights of the original code. 
\end{lemma}

This results apply directly to XYZ products codes when $-\1\not\in \mathcal{S}$. In the case where $-\1\in\mathcal{S}$ the initial stabilizer code with parameters $\llbracket n,k,d\rrbracket$ to consider is any given by a consistent choice of signs for the generators as explained in Sec.~\ref{sub:minusone}. The mapping to a CSS code with parameters $\llbracket 4n, 2k, 2d\rrbracket$ can be directly applied to the original set of generators as the mapping forgets phases.

\subsection{Comparison with 3-fold Hypergraph Product codes}
\label{sub:comparison}
Other approaches have used the structure given by $\mathbb{F}_2$ chain complexes, including ones build from $D$-fold hypergraph products, in order to build quantum stabilizer codes such as D-dimensional color codes \cite{BM07} or quantum pin codes \cite{VB19} but our approach is completely different here.
The 3-fold hypergraph product code construction described for instance in Refs.~\cite{ZP19,ZP20,QVR20} is maybe the closest to the XYZ product code construction and we now explain how they differ.

The 3-fold hypergraph product code makes use of the same chain complex structure as described in Figure~\ref{fig:chaincomplex}, although not in the same way.
First it discards one of the ends of the complex, say $A$, and associates qubits with only one level of the complex: the blocks $B$, $C$ and $D$ here.
Then one uses the blocks $S$, $T$ and $U$ as $Z$-type stabilizers and block $V$ as $X$-type stabilizers with exactly the same incidence structure but therefore different Pauli operators.

One can also take as an example the difference between the 3D toric code and the Chamon code which are instances of 3-fold hypergraph products and XYZ product codes respectively when the three classical codes are repetition codes.
In that case the chain complex is most conveniently described by a 3D cubic lattice where $A$ are the vertices; $S$, $T$ and $U$ are the edges in the $x$, $y$ and $z$ directions respectively; $B$, $C$ and $D$ are the faces perpendicular to the $x$, $y$ and $z$ directions respectively; and $V$ are the cubes.
This is schematically represented in Figure \ref{fig:3DToricandChamon}.
On this structure the 3D toric code consists in putting qubits on the faces, $Z$-checks on the four faces surrounding each edge and $X$-checks on the six faces surrounding each cube, see Figure~\ref{fig:3DToricandChamon} on the left.
On the other hand, the Chamon code consists in putting qubits on the faces as well as on the vertices and having checks defined by the edges and cubes in the following way. 
Each cube, or edge respectively, is surrounded by six faces, or four faces and two vertices respectively.
The checks act with the Pauli $X$, $Y$ or $Z$ respectively, on the elements away from them in the $x$, $y$ or $z$ direction respectively, see Figure~\ref{fig:3DToricandChamon} on the right.
\begin{figure}[!h]
	\centering
	\includegraphics[width=0.35\linewidth]{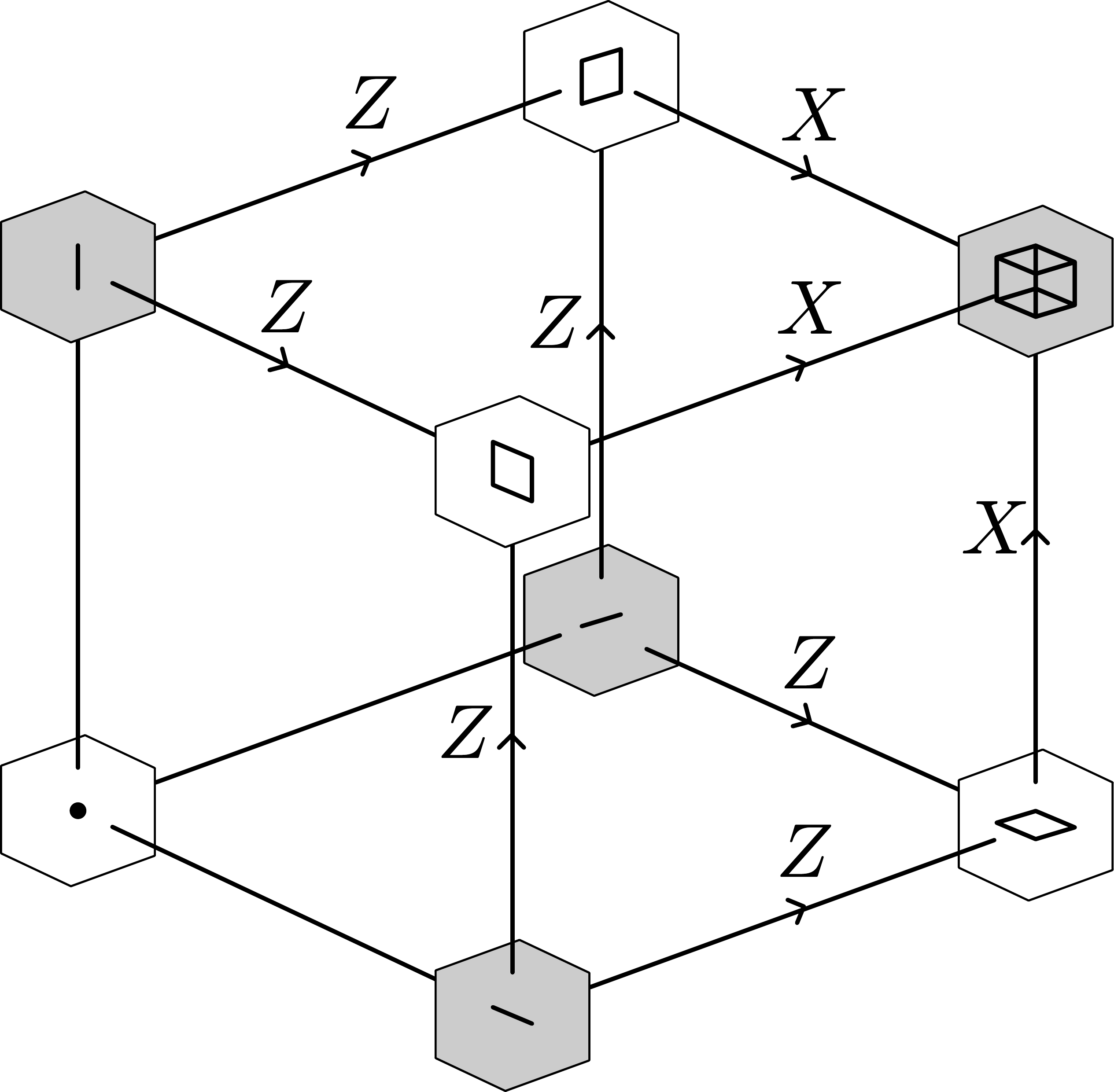}\hfil	\includegraphics[width=0.35\linewidth]{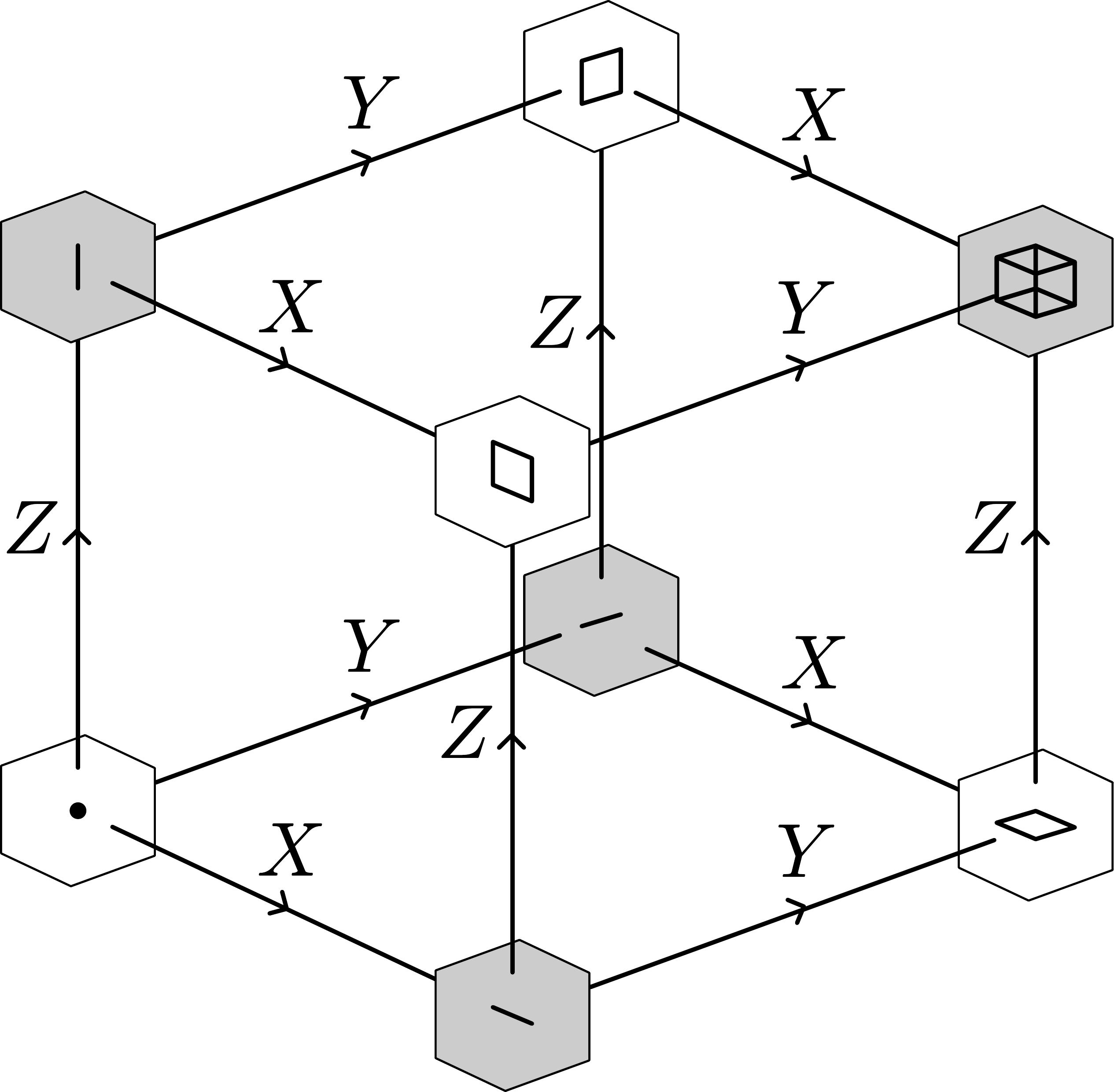}
	\caption{Geometrical interpretation of the blocks when the three classical codes are repetition codes as well as the Pauli operator chosen for the stabilizer of the 3D toric code left and Chamon code right.}  
	\label{fig:3DToricandChamon}
\end{figure}

When the three classical codes are local in 1D then both construction will therefore be local in 3D, using for example this cubic lattice to lay out the qubits. We will study these configurations in detail in Sections \ref{sec:cyclic} and \ref{sec:simplest}.

\section{Dimension of the XYZ product codes}
\label{sec:dim}

In this section, we compute the dimension of the XYZ product codes by finding the number of independent generators. As mentioned before, the XYZ product code is readily well defined if it can be shown that $-\1$ does not belong to the stabilizer group.
If it were the case that $-\1$ is in the group then the following result would apply to the stabilizer code obtained from consistently choosing the signs of the generators to turn it into a stabilizer group (see Section \ref{sub:minusone}) or with a factor $2$ to the mapped CSS version of the code (see Section \ref{sub:CSS}).

If $r$ is the number of ``relations'' between the generators, then the number of independent generators is $n_S + n_T + n_U + n_V - r$.
The dimension $k$ of the code is then given by the formula
\begin{align*}
k
&= n_A + n_B + n_C + n_D - n_S - n_T - n_U - n_V + r \\
&= (n_1 - m_1)(n_2 - m_2)(n_3 - m_3) + r.
\end{align*}
Specific to the XYZ product code is the occurrence of nontrivial relations between generators of different type.
This is because there are generators of $\sigma_1$, $\sigma_2$ as well as $\sigma_3$ type, and we have that $\sigma_1 \sigma_2 \sigma_3 \propto \1$.
This does not occur in CSS constructions.
Specifically we find a relation between a set of generators if the $\sigma_1$-, the $\sigma_2$- and the $\sigma_3$-checks coincide in each of the blocks.
This corresponds to a solution $(S,T,U,V)$ of the ``Pauli system''
\[
\begin{bmatrix}
\uX^\dagger & \uY^\dagger & \uZ^\dagger & \cdot\\
\uY & \uX &\cdot& \uZ^\dagger \\
\uZ & \cdot& \uX & \uY^\dagger \\
\cdot& \uZ & \uY &\uX^\dagger \end{bmatrix}
\begin{bmatrix} S\\T \\U \\V\end{bmatrix}
=
\begin{bmatrix}
(\uX^\dag S) (\uY^\dag T) (\uZ^\dag U) \\
(\uY S) (\uX T) (\uZ^\dag V) \\
(\uZ S) (\uX U) (\uY^\dag V) \\
(\uZ T) (\uY U) (\uX^\dag V)
\end{bmatrix}
\propto \1.
\]
Now consider the shorthand notation $H_1 = H_1 \otimes \1_{\ell_2} \otimes \1_{\ell_3}$, $H_2 = \1_{\ell_1} \otimes H_2 \otimes \1_{\ell_3}$ and $H_3 = \1_{\ell_1} \otimes \1_{\ell_2} \otimes H_3$, where $\ell_i \in \{n_i,m_i\}$ will be chosen so that the dimensions match, similar to the previous section.
Then \textit{e.g.}~the condition $(\uX^\dag S) (\uY^\dag T) (\uZ^\dag U) \propto \1$ becomes equivalent to $H_1^T S = H_2^T T = H_3^T U$.
By the same argument, we can reexpress the ``Pauli system'' using the linear system (over $\mathbb{F}_2$)
\begin{equation} \label{eq:system-dim}
\begin{alignedat}{3}
H_1^T S &= H_2^T T &&= H_3^T U \\
H_1 T &= H_2 S &&= H_3^T V \\
H_1 U &= H_2^T V &&= H_3 S \\
H_1^T V &= H_2 U &&= H_3 T
\end{alignedat}
\end{equation}
We only managed to determine the number of solutions to this linear system (and hence the code dimension) under certain constraints on the parity-check matrices.
The following theorem is a special case (which follows from combining Corollary \ref{cor:dim} with Theorem \ref{thm:syl-3}), and it nicely demonstrates the algebraic nature of the code family.
We let $p^\ell_i$ denote the characteristic polynomial of the $i$-th Jordan block (over the algebraic closure $\overline{\mathbb{F}}_2$) of $H_\ell H_\ell^T$, for $\ell=1,2,3$.
\begin{reptheo}{thm:xyz-dim}
If the parity-check matrices $H_1$, $H_2$ and $H_3$ are invertible, then the dimension of the XYZ product code is equal to the number of solutions of the tensor Sylvester equation
\[
H_1 H_1^T X
= H_2 H_2^T X
= H_3 H_3^T X,
\]
which is equal to
\[
\sum_{i,j,k} \mathrm{deg}(\mathrm{gcd}(p^1_i,p^2_j,p^3_k)).
\]
\end{reptheo}

Using the more elementary notion of \textit{similarity invariant} from abstract algebra \cite{serre2002matrices}, this can be equivalently expressed as
\[
\sum_{i,j,k} \mathrm{deg}(\mathrm{gcd}(h^1_i,h^2_j,h^3_k)),
\]
where $h^\ell_1,\dots,h^\ell_n$ are the similarity invariants of $H_\ell H_\ell^T$ for $\ell=1,2,3$.
From Theorem \ref{thm:xyz-dim} we get the following corollary, which provides a useful condition for the code to have dimension 1.
\begin{corol}\label{cor:T}
If the parity-check matrices $H_1$, $H_2$ and $H_3$ are invertible and the matrices $H_1 H_1^T, H_2 H_2^T$ and $H_3 H_3^T$ have a unique common eigenvalue 1, each with geometric multiplicity 1, then the XYZ product code has dimension 1.
\end{corol}

Before going into the proofs, we comment on some obstacles towards a full solution of the dimension question. We feel that also in the general, non-invertible case the dimension will be largely determined by the tensor Sylvester equation. Indeed, if we consider \eqref{eq:system-dim} over the complex numbers instead of over $\mathbb{F}_2$, then we can use the Moore-Penrose pseudo-inverse in the proof of Claim \ref{claim:dim} to show that a solution to the Sylvester equation always implies a solution to the full system. Over $\mathbb{F}_2$, however, there does not exist in general a Moore-Penrose inverse, and the argument breaks down. This is closely tied to the fact that over $\mathbb{F}_2$ potentially $\mathrm{rank}(H_1) > \mathrm{rank}(H_1 H_1^T)$ (\textit{e.g.}, if $H_1 = [1 \, 1]$), see \cite{wu1998existence}. As a consequence, the Sylvester equation (involving terms $H_\ell H_\ell^T$) no longer captures the ``full picture''. Indeed, if say $H_1 = [1 \, 1], H_2=H_3=0$, then the (trivial) Sylvester equation has a nonzero solution whereas the original system does not. In some sense, this suggests that homological tools (in addition to the algebraic ones) might be useful. Indeed, the extremal case $H_1 H_1^T = 0$ for $H_1 \neq 0$ describes precisely the boundary condition of a chain complex, so that the rank deficiency of $H_1 H_1^T$ might signal nontrivial homological properties.\\

Let us also mention here what needs to be proven to establish that the non-CSS version of an XYZ product code is a valid stabilizer code. As already pointed out, this would follow from the fact that $-\1$ does not belong to the stabilizer group. 
Any relation among the generators corresponds to a choice of $(S,T,U,V)$ such that \eqref{eq:system-dim} holds. Assuming such a choice and denoting by $n_A, n_B, n_C, n_D$ the Hamming weight of the tensors appearing in each of the 4 equations and $g_S, g_T, g_U, g_V$ the products of the generators of type $S, T, U, V$ among the relation, we want to check that 
\[ g_S g_T g_U g_V = \1.\]
Note that such a relation involves $n_A$ products of type $\sigma_1 \sigma_2 \sigma_3 = i\1$ in $A$, $n_B$ products of type $\sigma_2 \sigma_1 \sigma_3 = -i \1$ in $B$, $n_C$ products of type $\sigma_3 \sigma_1 \sigma_2 = i\1$ in $C$ and $n_D$ products of type $\sigma_3 \sigma_2 \sigma_1 = -i\1$ in $D$. 
The global phase of $g_S g_T g_U g_V $ is therefore
\[ i^{n_A - n_B + n_C - n_D}\]
which will be equal to 1 if and only if
\begin{align}\label{eq:suff-1}
n_A - n_B+n_C - n_D = 0 \mod 4.
\end{align}
Establishing that this holds for any relation among the generators would imply that the XYZ product code is a valid stabilizer code.

\subsection{Code dimension and a tensor Sylvester equation}
\label{sub:sylvester}
The algebraic nature of the code becomes apparent when trying to analyze the linear system in Eqn.~\eqref{eq:system-dim}.
In particular, we get the following necessary condition on any solution of the system, which is a variant on the (homogeneous) Sylvester equation.
Similar equations hold for the other tensors.
\begin{lemma}[Tensor Sylvester Equation]
If $(S,T,U,V)$ is a solution of \eqref{eq:system-dim} then necessarily
\begin{equation} \label{eq:sylv}
H_1 H_1^T V
= H_2 H_2^T V
= H_3 H_3^T V.
\end{equation}
\end{lemma}
\begin{proof}
If $(S,T,U,V)$ is a solution of \eqref{eq:system-dim}, then $H_1 H_1^T V = H_1 H_2 U = H_2 H_1 U = H_2 H_2^T V$.
Similarly, $H_2 H_2^T V = H_2 H_3 S = H_3 H_2 S = H_3 H_3^T V$.
\end{proof}

Now let $s$ denote the number of independent nonzero solutions to \eqref{eq:sylv}.
We can show that, under additional conditions on the parity-check matrices, we can get an explicit expression for the number of solutions as a function of $s$.
We denote $k_\ell = \mathrm{dim}(\ker(H_\ell))$ and $k^T_\ell = \mathrm{dim}(\ker(H_\ell^T))$ for $\ell=1,2,3$.

\begin{claim} \label{claim:dim}
If $H_2 H_2^T$ and $H_3 H_3^T$ are invertible, then the number of independent nonzero solutions to \eqref{eq:system-dim} is given by
\[
s + k^T_1 k_2 k_3.
\]
\end{claim}
\begin{proof}
We already showed that \eqref{eq:sylv} is a necessary condition.
However, it is also a sufficient condition on $V$ since any solution $V$ to \eqref{eq:sylv} describes a solution to \eqref{eq:system-dim} by setting
\[
S = H_3^T (H_3 H_3^T)^{-1} H_2^T V, \quad
T = H_3^T (H_3 H_3^T)^{-1} H_1^T V, \quad
U = H_2^T (H_2 H_2^T)^{-1} H_1^T V,
\]
which is straightforward to verify.
Now assume that we have an alternative solution of the form $(S+S',T+T',U+U',V)$, then necessarily $S' \in \ker(H_1^T) \cap \ker(H_2) \cap \ker(H_3)$, $T' \in \ker(H_1) \cap \ker(H_2^T) \cap \ker(H_3) = \varnothing$ and $U' \in \ker(H_1) \cap \ker(H_2) \cap \ker(H_3^T) = \varnothing$ (using that by our assumption $H_2^T$ and $H_3^T$ have full rank).
The claim follows by noting that $\mathrm{dim}(\ker(H_1^T) \cap \ker(H_2) \cap \ker(H_3)) = k^T_1 k_2 k_3$ as a consequence of the tensor structure.
\end{proof}
\noindent
By permutation symmetry, we reach a similar conclusion if the pairs $H_1 H_1^T$ and $H_2 H_2^T$ or $H_1 H_1^T$ and $H_3 H_3^T$ are invertible.
In fact, by considering the Sylvester equation on $S$, $T$ or $U$ instead of on $V$, we get any of the 12 possibilities (invertibility of pairs of the form $H_\ell H_\ell^T$ or $H_\ell^T H_\ell$).
We have the following corollary.
\begin{corol} \label{cor:dim}
If $H_2 H_2^T$ and $H_3 H_3^T$ are invertible, then the dimension of the code is given by
\[
(n_1 - m_1)(n_2 - m_2)(n_3 - m_3) + s + k_1^T k_2 k_3.
\]
\end{corol}

\subsection{A tensor Cecioni-Frobenius theorem}
In this section we will determine the number of solutions to the tensor Sylvester equation \eqref{eq:sylv} over some field $K$.
To this end we derive a tensor generalization of the \textit{Cecioni-Frobenius theorem} which concerned the (homogeneous) Sylvester equation $AX = XB$.

For square matrices $A \in K^{n \times n}$, $B \in K^{m \times m}$ and $C \in K^{\ell \times \ell}$, consider the shorthand $A = A \otimes \1_m \otimes \1_\ell$, $B = \1_n \otimes B \otimes \1_\ell$, $C = \1_n \otimes \1_m \otimes C$.
We define the tensor Sylvester equation (for 3-tensors) as
\begin{equation} \label{eq:syl-3}
A X
= B X
= C X,
\end{equation}
where $X \in K^{m \times n \times \ell}$.
We are interested in the number of independent solutions of this equation.
We will prove the following tensor Cecioni-Frobenius theorem, where $p^A_i,p^B_i,p^C_i$ denote the characteristic polynomials of the $i$-th Jordan block of $A,B,C$, respectively (over the algebraic closure $\overline K$).
\begin{theo}[Tensor Cecioni-Frobenius] \label{thm:syl-3}
The number of independent solutions of $AX = BX = CX$ is
\[
\sum_{i,j,k=1}^n \mathrm{deg}(\mathrm{gcd}(p^A_i,p^B_j,p^C_k)).
\]
\end{theo}
\noindent
Our proof runs somewhat parallel to the analysis in \cite{gantmacher1959theory} for the regular Cecioni-Frobenius theorem.
We start with the following elementary fact:
\begin{lemma} \label{lem:lin-system-closure}
The number of independent solutions of a linear system over some field $K$ is equal to the number of independent solutions of the linear system considered over the algebraic closure $\overline K$.
\end{lemma}
\begin{proof}
The number of independent solutions of a linear system $A x = b$ over some field $K$ is equal to the dimension of the kernel of the augmented matrix $[A|b]$ (Rouch\'e-Capelli theorem), which is equivalent for any field containing the coefficients of $A$ and $b$.
One way of seeing this is that the rank follows from reducing $[A|b]$ to its reduced row echelon form, and all operations in this transformation can be performed over the original field.
\end{proof}

By this lemma we can prove the tensor Cecioni-Frobenius theorem over a field $K$ by proving it over the algebraic closure $\overline K$.
This is helpful because over $\overline K$ the matrices $A,B,C$ are similar to their Jordan normal forms: \textit{i.e.}, there exist invertible $P_A,P_B,P_C$ such that
\[
A
= P_A \Big( \bigoplus_i J^A_i \Big) P_A^{-1}, \quad
B
= P_B \Big( \bigoplus_j J^B_j \Big) P_B^{-1}, \quad
C
= P_C \Big( \bigoplus_k J^C_k \Big) P_C^{-1},
\]
where $\{J^A_i\},\{J^B_j\},\{J^C_k\}$ correspond to the Jordan blocks of $A$, $B$, $C$, respectively.
Through the change of variables $X = (P_A \otimes P_B \otimes P_C) Y$, the Sylvester equation \eqref{eq:syl-3} becomes equivalent to
\[
\big( \bigoplus_i J^A_i \big) \, Y
= \big( \bigoplus_j J^B_j \big) \, Y
= \big( \bigoplus_k J^C_k \big) \, Y.
\]
Now partition $Y$ into blocks $Y_{ijk}$ according to the ranges of the Jordan blocks of $A,B,C$.
More specifically, let the intervals $n_i^A \subset [n]$, $n_j^B \subset [m]$, $n_k^C \subset [\ell]$ denote to the rows of $J_i^A$, $J_j^B$ and $J_k^C$, respectively (as part of the larger matrices $\bigoplus_i J_i^A$, etc).
Then the block $Y_{ijk}$ follows by restricting the first index of $Y$ to $n_i^A$, the second index to $n_j^B$ and the third index to $n_k^C$.
The equation then splits up into blocks: we effectively need that
\[
J^A_i \, Y_{ijk}
= J^B_j \, Y_{ijk}
= J^C_k \, Y_{ijk},
\quad \forall i,j,k.
\]
This reduces the problem to the tensor Sylvester equation for Jordan blocks.

Let $J(\lambda,\alpha)$ denote a Jordan block of eigenvalue $\lambda$ and dimension $\alpha$, and consider the equation
\begin{equation} \label{eq:syl-jordan}
J(\lambda,\alpha) X
= J(\mu,\beta) X
= J(\nu,\gamma) X,
\end{equation}
where we again omit tensor notation.
\begin{lemma}
Equation \eqref{eq:syl-jordan} has no nonzero solution unless $\lambda = \mu = \nu$.
\end{lemma}
\begin{proof}
Assume that there exists a nonzero solution $X \neq 0$.
Note that $A X = B X = C X$ implies that $p(A) X = p(B) X = p(C) X$ for any polynomial $p$.
Indeed, if say $A X = B X$ then $A^2 X = A B X = B A X = B^2 X$, and by induction $A^k X = B^k X$ for any integer $k$ (and hence $p(A) X = p(B) X$ for any polynomial $p$).
The same reasoning implies that $p(A) X = p(C) X$.

Now if we let $p_A$ be the characteristic polynomial of $A$, then
\[
p_A(A) X
= p_A(B) X
= p_A(C) X
= 0.
\]
Since $X \neq 0$ and $p_A = (x-\lambda)^\alpha$, this implies that $\det(p_A(B)) = 0 = \det(B-\lambda I)^\alpha$.
$B$ must hence have an eigenvalue $\lambda$, and so $\mu = \lambda$.
An analogous reasoning proves that $\nu = \lambda$.
\end{proof}

This shows that we should only consider the special case $J(\lambda,\alpha) X = J(\lambda,\beta) X = J(\lambda,\gamma) X$.
Since $J(\lambda,\alpha) = \lambda I_\alpha + J(0,\alpha)$, this is equivalent to
\begin{equation} \label{eq:syl-jordan-hom}
J(0,\alpha) X
= J(0,\beta) X
= J(0,\gamma) X.
\end{equation}
\begin{lemma}
The number of independent solutions of \eqref{eq:syl-jordan-hom} is $\min\{\alpha,\beta,\gamma\}$.
\end{lemma}
\begin{proof}
Without loss of generality, assume that $\gamma = \min\{\alpha,\beta,\gamma\}$.
The tensor operators $J(0,\alpha) \otimes I \otimes I$, $I \otimes J(0,\beta) \otimes I$ and $I \otimes I \otimes J(0,\gamma)$ effectively work as $x$-, $y$- and $z$-\textit{shifts} on the entries of $X$, respectively.
Indeed, componentwise the equation becomes
\begin{equation} \label{eq:plane}
X_{i+1,j,k}
= X_{i,j+1,k}
= X_{i,j,k+1},
\quad \forall i \in [\alpha], j \in [\beta], k \in [\gamma],
\end{equation}
where we set $X_{i,j,k} = 0$ whenever $i > \alpha$, $j > \beta$ or $k > \gamma$.

Now pick any $(i,j,k)$ and set $x = X_{i,j,k}$.
By \eqref{eq:plane} we get that $X_{i',j',k'} = x$ for every $(i',j',k')$ such that $i'+j'+k' = i+j+k$.
Any solution can hence be described by a (linear combination of) uniform planes:
\[
X_{i,j,k}
= x_{i+j+k},
\]
for some set of values $\{x_\ell\}_{\ell \geq 3}$.
Moreover, we must have that $x_\ell = 0$ whenever $\ell > \gamma+2$, since we set $X_{1,1,k} = 0$ whenever $k>\gamma$.
This implies that the number of independent solutions is $\gamma = \min\{\alpha,\beta,\gamma\}$.
\end{proof}

Putting everything together, we decomposed any solution $Y$ as a direct sum of blocks $Y_{ijk}$.
A block $Y_{ijk}$ can be nonzero only if the Jordan blocks $J_i^A,J_j^B,J_k^C$ correspond to the same eigenvalue.
With $\alpha_i,\beta_j,\gamma_k$ the dimensions of these blocks, respectively, the number of independent solutions for $Y_{ijk}$ in this case is $\min\{\alpha_i,\beta_j,\gamma_k\} = \mathrm{deg}(\mathrm{gcd}(J_i^A,J_j^B,J_k^C))$.
Summing this up over all blocks, we get that the total number of independent solutions is
\[
\sum_{i,j,k}
\mathrm{deg}(\mathrm{gcd}(J_i^A,J_j^B,J_k^C)).
\]
This proves the tensor Cecioni-Frobenius theorem.

\subsection{Dimension of the modified Chamon code}
The Chamon code on a lattice of size $n_1 \times n_2 \times n_3$ has dimension $4 \gcd(n_1,n_2,n_3)$ \cite{BLT11}.
It is described by the XYZ product code with circulant matrices $H_i = \1 + \Omega_{n_i}$, for $i = 1,2,3$, where again $\Omega_{n_i} \ket{x} = \ket{x + 1 \mod n_i}$.
If we are interested in the minimum distance of the Chamon code, then in principle we have to lower bound the weight of all $4^4-1 = 255$ logical operators.
Here we present a slight modification of the code which results in dimension 1 (and hence only 3 logical operators).
Specifically, we will consider the $(n_i-1) \times n_i$ parity-check matrices $H_i$ obtained by deleting the last row of $\1 + \Omega_{n_i}$, for $i = 1,2,3$:
\[
[H_i]_{k,\ell}
= [\1 + \Omega_{n_i}]_{k,\ell} \quad \text{for} \quad k \in [n_i-1],\; \ell \in [n_i].
\]
We refer to the resulting XYZ code $\mathcal{Q}(H_1, H_2, H_3)$ as the modified Chamon code.
We can prove the following lemma.
\begin{lemma} \label{lem:mod-chamon}
For $H_1, H_2, H_3$ defined as above,
\begin{itemize}
\item
$H_i^T$ has rank $n_i - 1$.
\item
If $n_i$ is odd then $H_i H_i^T$ is invertible.
\item
If $\gcd(n_i,n_j) = 1$ then the intersection of the spectra of $H_i H_i^T$ and $H_j H_j^T$ are disjoint.
\end{itemize}
\end{lemma}
\begin{proof}
For the first bullet, note that the first $n_i-1$ rows of $H_i^T$ describe a Jordan block of size $n_i-1$ and eigenvalue 1.
Such a Jordan block has rank $n_i-1$, and therefore $H_i^T$ has rank $n_i-1$.
For the second and third bullets, note that $H_i H_i^T = \omega_{n_i-1} + \omega^T_{n_i-1}$, where $\omega_k$ is the $k \times k$ matrix with 1 on its upper diagonal 0 elsewhere.
Now denote the characteristic polynomial of $H_i H_i^T$ as $p_{n_i}(x) = \mathrm{det}(x \1_{n_i-1} + \omega_{n_i-1} + \omega_{n_i-1}^T)$ (for $n_i \geq 2$).
If we expand the determinant along the first row, we see that
\[
p_{n+1}(x)
= x p_n(x) + p_{n-1}(x).
\]
This recurrence holds for all $n \geq 1$ if we set $p_0(x) = 0$ and $p_1(x) = 1$.
The solution to this recurrence relation are the so-called \textit{Fibonacci polynomials} $p_{n+1}(x) = F_{n+1}(x)$.
As a consequence, the characteristic polynomial of $H_i H_i^T$ is exactly the Fibonacci polynomial $F_{n_i}(x)$.
A property of these polynomials is that $F_k$ divides $F_\ell$ over $\F_2$ if and only if $k$ divides $\ell$.
In particular, this implies that if $n_i$ is odd ($\gcd(n_i,2)=1$) then $F_2(x) = x$ does not divide the characteristic polynomial of $H_i H_i^T$, and hence the matrix is invertible.
In addition, $\gcd(F_k,F_\ell) = 1$ whenever $\gcd(k,\ell) = 1$, and this completes the proof.
\end{proof}

\begin{corol}
If $n_2$ and $n_3$ are odd and $\gcd(n_2,n_3) = 1$, then the modified Chamon code has dimension 1.
\end{corol}
\begin{proof}
From Section \ref{sec:dim} we know that the dimension of the XYZ product code is given by
\[
(n_1 - m_1)(n_2 - m_2)(n_3 - m_3) + r,
\]
with $r$ the number of solutions to system \eqref{eq:system-dim}.
For the modified Chamon code this becomes $1 + r$.
Since $H_2 H_2^T$ and $H_3 H_3^T$ are invertible, we can use Claim \ref{claim:dim}, which says that
\[
r
= s + k_1^T k_2 k_3,
\]
with $s$ the number of solutions to the tensor Sylvester equation \eqref{eq:sylv}, and $k_i,k_i^T$ the dimension of the kernel of $H_i,H_i^T$, respectively.
From Lemma \ref{lem:mod-chamon} we have that $k_1^T = 0$, so that $r = s$.
By the tensor Cecioni-Frobenius theorem (Theorem \ref{thm:syl-3}) we get that $s = 0$ if there is no common eigenvalue among the matrices $H_1H_1^T$, $H_2 H_2^T$ and $H_3 H_3^T$.
This is ensured by the last bullet in Lemma \ref{lem:mod-chamon}.
\end{proof}

\section{Minimum distance of the XYZ product code}
\label{sec:decoupling}

In this section, we study the minimum distance of XYZ product codes that encode a single logical qubit. 
While some of these results (such as the lower bound presented in Section \ref{sub:lower-dmin}) might hold in more generality, they are much easier to show in the case where there is a single logical qubit. The reason for this is that if there is a single logical qubit, there are only three logical operators ($X$, $Y$ or $Z$-logical operators) whose weight we need to bound to get the distance.
To further simplify things, we will restrict to triples $H_1, H_2, H_3$ in some particular set $\cal T$ of interest:
\begin{defn}
The set of triples $\mathcal{T}$ contains all triples $(H_1,H_2,H_3)$ of invertible matrices of odd dimension that uniquely fix the all-ones vector $u = (1,1,\ldots, 1)^T$, and such that the intersection of the spectra of $H_1 H_1^T, H_2 H_2^T$ and $H_3 H_3^T$ is otherwise empty.
\textit{I.e.},
\[
H_i \in GL(n_i, \F_2),\quad
\gcd(n_i,2)=1,\quad
H_i x = H_i^T x = x \Leftrightarrow x = u,\quad
\bigcap_{i=1}^3 \mathrm{Spec}(H_i H_i^T) = \{1\}.
\]
\end{defn}
By Corollary \ref{cor:T} the resulting codes have dimension equal to 1.
We will show how to reduce the problem of computing their minimum distance (up to constant factors) to computing the optimum value of a simple combinatorial problem over a binary tensor of order 3 (see Theorem \ref{thm:decoupling}).

Before that, we show that XYZ product codes obtained from triples on $\mathcal{T}$ are well-defined stabilizer codes. 
\begin{lemma} \label{lem:-1forT} The stabilizer group associated with the XYZ product code of $(H_1, H_2, H_3) \in \mathcal{T}$ does not contain $-\1$. 
\end{lemma}
\begin{proof}
In this setup, there is a single relation among the generators: the product of all the generators is proportional to the identity operator. In particular, the number of occurrences of products $\sigma_1 \sigma_2 \sigma_3$ is equal to $n_1 n_2 n_3$ in each of the sets $A, B, C, D$, showing that Eqn.~\eqref{eq:suff-1} is satisfied and therefore the stabilizer does not contain $-\1$.
\end{proof}

\subsection{Logical operators}

The logical operators of an XYZ product code ${\cal Q}(H_1,H_2,H_3)$ with $(H_1,H_2,H_3) \in {\cal T}$ take a particularly simple form.
With $u = (1,1,\dots,1)^T$ and $e_1 = (1,0,\dots,0)^T$, they can be represented as
\[
{\cal X}
= \begin{bmatrix}
\sigma_1^{e_1 \otimes u \otimes u} \\
0 \\ 0 \\
\sigma_1^{e_1 \otimes u \otimes u}
\end{bmatrix}, \qquad
{\cal Y}
= \begin{bmatrix}
\sigma_2^{u \otimes e_1 \otimes u} \\
0 \\
\sigma_2^{u \otimes e_1 \otimes u} \\
 0
\end{bmatrix}, \qquad
{\cal Z}
= \begin{bmatrix}
\sigma_3^{u \otimes u \otimes e_1} \\
\sigma_3^{u \otimes u \otimes e_1} \\
0 \\ 0
\end{bmatrix}.
\]
This corresponds to pairs of slices of $\sigma_i$ operators.

\begin{figure}[!h]
\centering
\includegraphics[width=0.4\linewidth]{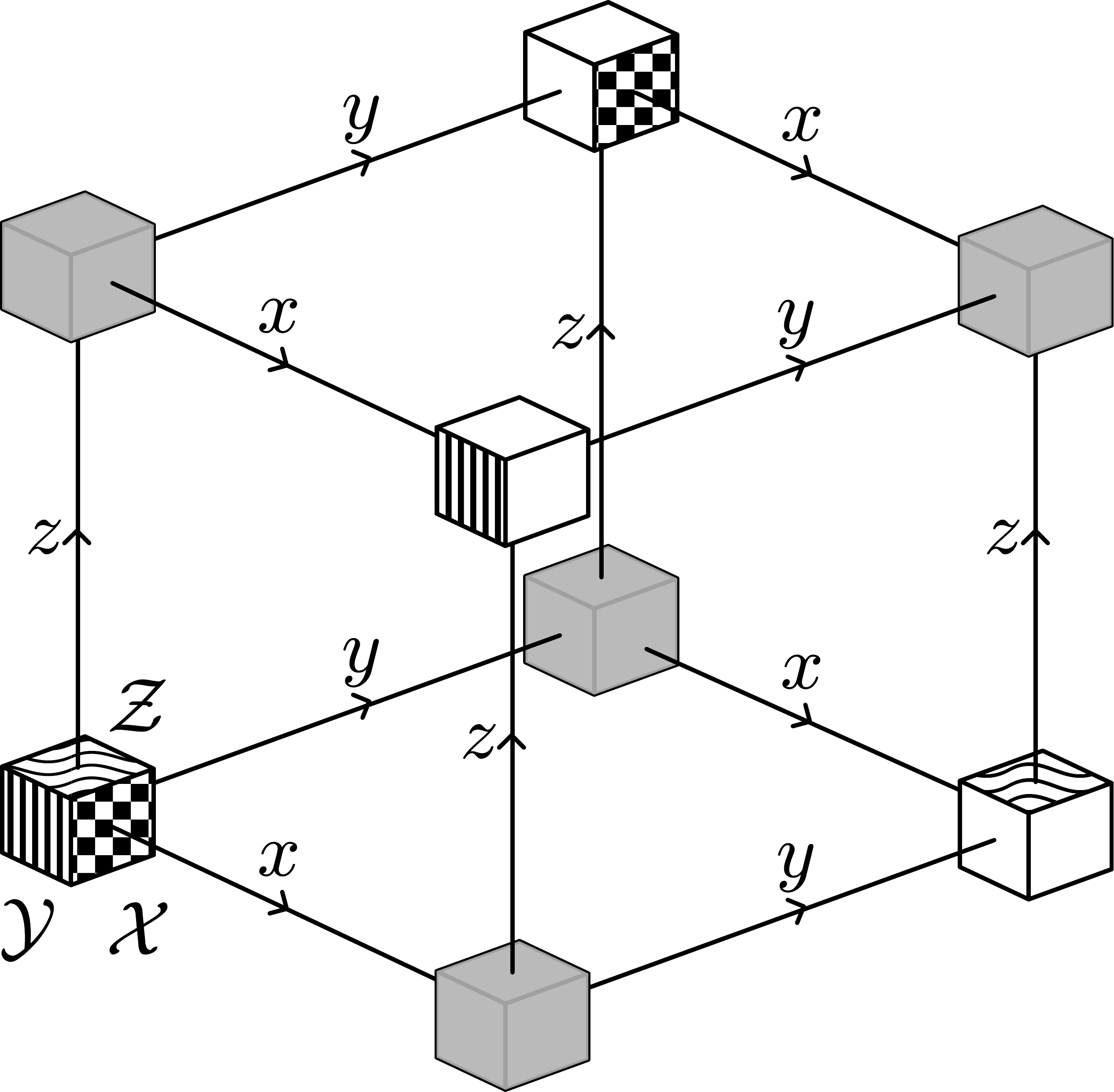}
\caption{Representation of the logical errors. For clarity of exposition, slight variants on the logical $\cal X$ and $\cal Z$ are shown (exchanging $e_1$ for $e_{n_i}$ in the definition).}
\label{fig:logicals}
\end{figure}

\noindent
To see that they are fixed by the stabilizer group, it suffices to verify that
\[
\Sigma * {\cal X}
=
\begin{bmatrix}
\uX * \sigma_1^{e_1 \otimes u \otimes u} \\
\uY * \sigma_1^{e_1 \otimes u \otimes u} + \uZ^\dag * \sigma_1^{e_1 \otimes u \otimes u} \\
\uZ * \sigma_1^{e_1 \otimes u \otimes u} + \uY^\dag * \sigma_1^{e_1 \otimes u \otimes u} \\
\uX * \sigma_1^{e_1 \otimes u \otimes u}
\end{bmatrix}
=
\begin{bmatrix}
0 \\
(\sigma_2 \ast \sigma_1)^{e_1 \otimes (H_2 u) \otimes u} + (\sigma_3 * \sigma_1)^{e_1 \otimes u \otimes (H_3^T u)} \\
(\sigma_3 * \sigma_1)^{e_1 \otimes u \otimes (H_3 u)} + (\sigma_2 * \sigma_1)^{e_1 \otimes (H_2^T u) \otimes u} \\
0
\end{bmatrix}
 = 0,
\]
using that $H_2 u = H_2^T u = H_3 u = H_3^T u$.
The same argument applies for ${\cal Y}$ and ${\cal Z}$.
Next we verify that they anticommute pairwise.
To this end, notice that the intersection of any pair of logical operators corresponds to the intersection of exactly two perpendicular slices, which has Hamming weight $n_i$ for some $i$.
By our assumption that the $n_i$'s are odd, this implies that they anticommute.
Finally, we note that this implies that they do not belong to the stabilizer group, because any pair of elements in the stabilizer group commutes.

\subsection{First lower bound on the minimum distance}
\label{sub:lower-dmin}

Using a standard trick we can get a simple lower bound on the minimum distance when $(H_1,H_2,H_3) \in {\cal T}$.
Recall that we can prove a lower bound on the distance by proving a lower bound on the support of a logical error.

First we note that in fact the logical operator $\cal X$ has $2 n_1$ disjoint representatives:
\[
{\cal X}
\in \left\{
\begin{bmatrix}
\sigma_1^{e_i \otimes u \otimes u} \\
0 \\ 0 \\
\sigma_1^{e_i \otimes u \otimes u}
\end{bmatrix}, \;
\begin{bmatrix}
0 \\
\sigma_1^{e_i \otimes u \otimes u} \\
\sigma_1^{e_i \otimes u \otimes u} \\
0
\end{bmatrix}
: i = 1,2,\dots,n_1
\right\}.
\]
This corresponds to picking pairs of slices in $A$ and $D$ or $B$ and $C$, and translating them in the $x$-direction.
We get a similar set of representations for $\cal Y$ and $\cal Z$.

Now let $\cal E$ denote a logical error.
Then $\cal E$ has to anticommute (and hence intersect) with \textit{each} of the representatives of either $\cal X$, $\cal Y$ or $\cal Z$, since otherwise $\cal E$ is in the stabilizer group.
Now assume \textit{e.g.}~that it intersects with each of the $2 n_1$ representatives of $\cal X$.
Since these representatives are all disjoint, this implies that $\cal E$ must have Hamming weight at least $2 n_1$.
By a similar reasoning for $\cal Y$ and $\cal Z$, we get the following lemma.
\begin{lemma} \label{lem:standard}
If $(H_1,H_2,H_3) \in \cal T$, then the XYZ product code ${\cal Q}(H_1,H_2,H_3)$ has minimum distance
\[
d
\geq 2 \min\{n_1,n_2,n_3\}.
\]
\end{lemma}
\noindent
In fact, this bound is tight as soon as two of the three parity-check matrices coincide. 
\begin{lemma}
Let $H, H'$ be parity-check matrices, with $H$ square of dimension $n$. The minimum distance of the XYZ product code $\mathcal{Q}(H, H, H')$ is upper bounded by $n$. 
\end{lemma}
\begin{proof}
Consider the 3-dimensional tensor $\1 \otimes e_1$ corresponding to a single horizontal layer equal to the identity $\1_n$ and $(n_3-1)$ horizontal layers identically equal to zero: 
\[ \1 \otimes e_1 (i,j,k) := \delta_{i,j} \delta_{k,1}.\]
We take as a logical error
\[
\begin{bmatrix}
{\cal A} & {\cal B} & {\cal C} & {\cal D}
\end{bmatrix}^T
=
\begin{bmatrix}
\sigma_3^{\1 \otimes e_1}
& \sigma_3^{\1 \otimes e_1}
& 0 & 0
\end{bmatrix}^T,
\]
\textit{i.e.}, the product of $Z$-Pauli operators on cells of index $(i,i,1)$ for $i \in [n]$ for qubits both in $A$ and $B$.
This operator has intersection 1 (and hence anticommutes) with both $\cal X$ and $\cal Y$.
On the other hand, it is fixed by the stabilizer group:
\[
\Sigma *
\begin{bmatrix}
{\cal A} \\ {\cal B} \\ {\cal C} \\ {\cal D}
\end{bmatrix}
= \begin{bmatrix}
(\sigma_1 * \sigma_3)^{H \otimes e_1} + (\sigma_2 * \sigma_3)^{H \otimes e_1} \\
(\sigma_2 * \sigma_3)^{H^T \otimes e_1} + (\sigma_1 * \sigma_3)^{H^T \otimes e_1} \\
0 \\ 0
\end{bmatrix}
= 0. \qedhere
\]
\end{proof}

\subsection{A decoupling argument}
\label{subsec:decoupling}
If $n_1$, $n_2$ and $n_3$ are of the same order, then the previous argument gives a lower bound of the form $\Omega(n_i) = \Omega(N^{1/3})$.
Here we outline a strategy to proving a lower bound of the form $\Omega(N^{2/3})$.
We do so by minimizing the weight of the logical operators.
Without loss of generality, we can assume that the minimum weight of a logical operator is reached for a $Z$-logical operator.
This can always be enforced by suitably permuting the three parity-check matrices, and we will therefore need to consider all 6 permutations in the following.

Throughout the argument we will again assume that $(H_1,H_2,H_3) \in \cal T$, but also that the row and column weights of the $H_i$'s are bounded by some $w \in O(1)$ (which is of course necessary to obtain an LDPC code).
We will also write a Pauli tensor $\cA$ as 
\[
\cA
=
\begin{bmatrix} A_x & A_y & A_z\end{bmatrix}^T
\]
to mean that $\cA = (\sigma_X)^{A_x} (\sigma_Y)^{A_y} (\sigma_Z)^{A_z}$.
In other words, $A_x$ is a binary tensor and we apply a Pauli-$X$ on each element of the support of $A_x$. This decomposition is not unique: for instance, for any binary tensor $M$, we have 
\[
\begin{bmatrix} A_x & A_y & A_z\end{bmatrix}^T
=
\begin{bmatrix} A_x + M & A_y + M & A_z + M\end{bmatrix}^T.
\]

\begin{lemma} \label{lem:pauli}
The (Pauli) weight of the Pauli operator $\cA$ is given by
\[
|\cA|
= \frac{1}{2}( | A_x + A_y|  + | A_x + A_z| + | A_y + A_z|).
\]
\end{lemma}
\begin{proof}
Let subsets $S_x,S_y,S_z$ correspond to the supports of $A_x,A_y,A_z$, respectively.
Then the number of nontrivial Pauli operators is given by
\begin{align*}
|S_x \cup S_y \cup S_z| - |S_x \cap S_y \cap S_z|\
&= | S_x | + | S_y | + | S_z | - | S_x \cap S_y | - | S_x \cap S_z | - | S_y \cap S_z | \\
&= \frac{1}{2} ( | S_x \triangle S_y | + | S_x \triangle S_z | + | S_y \triangle S_z |),
\end{align*}
with $\triangle$ the symmetric difference.
The lemma follows by noting that \textit{e.g.}~$| S_x \triangle S_y | = | A_x + A_y |$.
\end{proof}

Now we can write the most general form of a $Z$-logical operator as the sum of our previously defined logical operator $\cal Z$ (two horizontal slices of Pauli-$Z$ in the sets $A$ and $B$), which we will here denote by $R$, and of an arbitrary stabilizer element parametrized by tensors $(S,T,U,V)$:
\begin{align}\label{eqn:logical}
\cA = \begin{bmatrix} x S \\ y T \\ z U + R \end{bmatrix} \quad
\cB = \begin{bmatrix} x T \\ y S \\ z V + R \end{bmatrix} \quad
\cC = \begin{bmatrix} x U \\ y V \\ z S \end{bmatrix} \quad
\cD = \begin{bmatrix} x V \\ y U \\ z T \end{bmatrix},
\end{align}
where we use the shorthand $x = H_1 \otimes \1 \otimes \1$, $y = \1 \otimes H_2 \otimes \1$ and $z = \1 \otimes \1 \otimes H_3$.
For simplicity we assume that the parity-check matrices $H_i$ are all symmetric, but we will later remove this assumption.
The minimum distance $d_{\min}$ of the code is now given by the minimum weight of this Pauli operator over all choices of $S,T,U,V$:
\[
d_{\min}
= \min_{S,T,U,V} | \cA | + | \cB | + | \cC | + | \cD |.
\]
Using Lemma \ref{lem:pauli} we can rewrite this as
\[
2 d_{\min}
= \min_{S,T,U,V} \sum_{E \in \{A, B, C, D\} }
	| E_{xy}| + | E_{xz}| + | E_{yz}|,
\]
where $A_{xy} := x S + y T, \, A_{xz} := x S + z U + R, \, A_{yz} := y T + z U + R$, etc. 
We will denote by
\[
{W}
= {W}(S,T,U,V)
:= \sum_{E \in \{A, B, C, D\} } | E_{xy}| + | E_{xz}| + | E_{yz}|,
\]
the quantity that we wish to minimize.

While we can independently optimize over the tensors $S,T,U,V$, the problem is that the individual terms (\textit{e.g.}~$| A_{xy} | = |x S + z U + R|$) depend on pairs of tensors.
We will now argue how to lower bound the sum by a sum of ``decoupled'' terms.
To start, consider the term $| A_{xz} | + | C_{xz}|$ of $W$.
By assumption, the row and column weights of the $H_i$'s (and hence the $x,y,z$ operators) are bounded by $w$.
This implies that $|x E| \leq w | E |$, for any tensor $E$.
Applied to $A_{xz}$ and $C_{xz}$, this gives:
\begin{align*}
| x S + z U + R |
& \geq \frac{1}{w} | x^2 S + x z U + x R |
= \frac{1}{w} |x^2 S + x z U + R |\\
|x U + z S|
&\geq \frac{1}{w} | x z U + z^2 S |,
\end{align*}
where we further exploited\footnote{As a side remark: this step is crucial for our purpose but it does not hold for the Chamon code \cite{cha05}. For that code, we would get $|x R| = O(n)$ and therefore this approach would not be able to prove any lower bound on the minimum distance of the Chamon code better than $\Omega(N^{1/3})$.} that $x R = R$.
A triangle inequality finally yields: 
\[
| A_{xz} | + | C_{xz}|
\geq | A_{xz} + C_{xz}|
= \frac{1}{w} | (x^2 + z^2) S + R |.
\]
As promised, this gives a term that only depends on a single tensor $S$.
A similar analysis allows to also decouple the following pairs:
\[
\begin{gathered}
| A_{yz} | + | D_{yz}| \geq \frac{1}{w} | (y^2 + z^2) T + R |, \;\;
| B_{xz} | + | D_{xz}| \geq \frac{1}{w} | (x^2+z^2) T + R |, \\
| B_{yz} | + | C_{yz}| \geq \frac{1}{w} | (y^2+z^2) S + R |.
\end{gathered}
\]
The remaining terms of $W$ can finally be bounded by
\[
| A_{xy} |+ | B_{xy} | \geq \frac{1}{2w} ( | (x^2+y^2) S | + | (x^2+y^2) T | ), \;\;
| C_{xy} |+ | D_{xy} | \geq \frac{1}{2w} ( | (x^2+y^2) U | + | (x^2+y^2) V | ).
\]
Putting it all together, we get that
\begin{align*}
W(S,T,U,V)
\geq &\frac{1}{w} \big( |(x^2 + z^2) S + R | + | (y^2+z^2) S + R | 
+ | (x^2+z^2) T + R | + | (y^2 + z^2) T + R | \big) \\
&+ \frac{1}{2w} \big( | (x^2+y^2) S | + | (x^2+y^2) T | + | (x^2+y^2) U | + | (x^2+y^2) V | \big).
\end{align*}
Now if we minimize the right hand side over $S,T,U,V$ independently, then we can simply set $U = V = 0$ and $S = T$ to get
\[
\min_{S,T,U,V} W(S,T,U,V)
\geq \min_s \frac{1}{w} ( 2|(x^2 + z^2) S + R | + 2| (y^2+z^2) S + R | + | (x^2+y^2) S | ).
\]
We have thus managed to decouple the four variables and reduced the optimization problem to a single (tensor) variable problem.
This proves the following lemma (after replacing the shorthands $x,y,z$ by their explicit forms).
\begin{lemma}
\label{lem:decoupling}
Let $\{ H_1, H_2, H_3\} \in \mathcal{T}$  such that $H_i = H_i^T$. Let $w$ be an upper bound on the row weight of $H_i$. The minimum distance of the XYZ product code $\mathcal{Q}(H_1,H_2,H_3)$ satisfies\begin{align*}
 d_{\min} &\geq \min_{(ijk)\in (123)} \min_M \frac{1}{w} \left( |(H_i^2+H_k^2)M + R|  +|(H_j^2+H_k^2)M + R| + \frac{1}{2} |(H_i^2+H_j^2) M|\right),\\
 &\geq \min_{(ijk)\in (123)} \min_M \frac{3}{4w} \left( |(H_i^2+H_k^2)M + R|  + |(H_i^2+H_j^2) M|\right)
\end{align*}
\end{lemma}
\noindent
The second inequality follows from a simple triangle inequality:
\begin{align*}
|(x^2+z^2)M &+ R| + |(x^2+z^2)M + R| \\
&\geq \frac{3}{4} |(x^2+z^2)M + R| + \frac{1}{4} (|(x^2+z^2)M + R| + |(y^2+z^2)M + R| ) \\
&\geq \frac{3}{4} |(x^2+z^2)M + R| + \frac{1}{4} (|(x^2+y^2)M| ) .
\end{align*}
The same analysis goes through essentially without modification in the non-symmetric case, where we get the bounds
\begin{align*}
| A_{xz} | + | C_{xz} | &\geq \frac{1}{w} | (x x^T + z^T z) s+r |, \qquad
| A_{yz} | + | D_{yz} | \geq \frac{1}{w} | (y y^T + z^T z) t+r |, \\
| B_{xz} | + | D_{xz} | &\geq \frac{1}{w}| (x^T x+z^T z) t + r |, \qquad
| B_{yz} | + | C_{yz} | \geq \frac{1}{w}| (y^T y+z^T z) s + r |,
\end{align*}
and we ignore the remaining terms in $W$.
\begin{lemma}[Non-symmetric version]
 The minimum distance of the XYZ product code associated with $(H_1,H_2,H_3)\in \mathcal{T}$ is bounded by
\begin{align*}
 d_{\min}
 \geq \min_{(ijk)\in (123)} \min_{S,T}
 \frac{1}{2w} \Big( &|(H_i H_i^T+H_k^T H_k)S + R|  +|(H_j^T H_j+H_k^TH_k)S + R| \\
 & \quad + |(H_i^T H_i+H_k^T H_k)T + R| + |(H_j H_j^T +H_k^TH_k)T + R| \Big).
 \end{align*}
\end{lemma}
\noindent
In the rest of this manuscript, we will focus on the symmetric case. 

\subsubsection*{Tightness of the Argument}
We now show that the argument above is essentially tight.
For a tensor $M$ (with dimensions equal to the tensor $S$), consider the choice
\[
S = xy M, \quad
T = x^2 M, \quad
U = yz M, \quad
V = xz M.
\]
With these, we get the Pauli tensors
\[
\cA
= \begin{bmatrix} x^2 y M\\ x^2 y M  \\ y z^2 M + R \end{bmatrix}, \qquad
\cB
= \begin{bmatrix} x^3 M\\ x y^2 M   \\ xz^2 M + R \end{bmatrix}, \qquad
\cC
= \begin{bmatrix} xyz M\\ xyzM    \\ xyz M \end{bmatrix}, \qquad
\cD
= \begin{bmatrix} x^2 z M \\ y^2z M   \\ x^2zM \end{bmatrix},
\]
which represents a valid logical $Z$ operator. 
Using Lemma \ref{lem:pauli} we can bound its weight as follows (again using that $x R = y R = R$):
\begin{align*}
|\mathcal{A}|
&= | y((x^2 + z^2)M + R)|
\leq w |(x^2 + z^2) M + R |\\
|\mathcal{B}|
&= \frac{1}{2}(| x( x^2+y^2) M| + |x((x^2+z^2)M + R)| + |x((y^2+z^2)M + R)|) \\
&\leq \frac{w}{2} \left(  | ( x^2+y^2) M| + |(x^2+z^2)M + R| + |(y^2+z^2)M + R|\right)\\
|\mathcal{C}|
& =0\\
|\mathcal{D}|
&= |z(x^2+y^2) M|
\leq w |(x^2+y^2) M|
\end{align*}
so that
\[
| \cA |+| \cB |+| \cC |+| \cD |
\leq \frac{w}{2} ( 3| ( x^2+y^2) M| + 3|(x^2+z^2)M + R| + |(y^2+z^2)M + R| ).
\]
This gives us an \textit{upper bound} on the minimum distance as a function of the tensor $M$.
If now we let $M$ be the tensor that minimizes the first inequality in Lemma \ref{lem:decoupling}, we get the following.
\begin{theo}[Decoupling Theorem] \label{thm:decoupling}
Let $(H_1, H_2, H_3) \in \mathcal{T}$ be symmetric matrices with row weight upper bounded by $w$.
Define $d^*$ as the solution of the following combinatorial problem:
\begin{align}\label{eqn:d*}
d^* :=  \min_{(ijk)\in (123)} \min_M  \left( |(H_i^2+H_k^2)M + R|  +|(H_j^2+H_k^2)M + R| + \frac{1}{2} |(H_i^2+H_j^2) M|\right).
\end{align}
Then the minimum distance $d_{\min}$ of the XYZ product code $\mathcal{Q}(H_1, H_2, H_3)$ satisfies
\[ \frac{d^*}{w} \leq d_{\min} \leq \frac{3}{2}w d^*.\]
\end{theo}
Note that defining $d^*$ with any two of the three terms in \eqref{eqn:d*} only changes the value by a constant factor.

\subsection{Some obstructions to a lower bound}

\subsubsection{Permutation invariance}
Here we show that the distance of an XYZ product code with $(H_1, H_2, H_3) \in \mathcal{T}$ is invariant upon arbitrary permutations of rows and columns of the parity-check matrices.
This indicates that we cannot to improve the distance of the code by considering \textit{e.g.}~random permutations of the rows or columns.
\begin{lemma}[Permutation invariance]\label{lem:perm-inv}
Let $\pi_i, \tau_i$ be permutations over $[n_i]$ for $i=1,2,3$, and consider any triple $( H_1, H_2, H_3 ) \in \mathcal{T}$.
Then the distance of codes $\mathcal{Q}(H_1,H_2,H_3)$ and $\mathcal{Q}(\pi_1 H_1\tau_1,\pi_2 H_2 \tau_2,\pi_3 H_3\tau_3)$ coincide.
\end{lemma}

\begin{proof}
First, note that if $(H_1, H_2, H_3) \in \mathcal{T}$, then $(\pi_1 H_1\tau_1,\pi_2 H_2 \tau_2,\pi_3 H_3\tau_3) \in \mathcal{T}$ since multiplying $H_i$ by permutation matrices leaves it invertible, with the same row and column weights.
Moreover, $(\pi_i H_i \tau_i)(\pi_i H_i \tau_i)^T = \pi_i H_i H_i^T \pi_i^T$ has the same spectrum as $H_i H_i^T$.
Assume without loss of generality that the minimum distance of $\mathcal{Q}(H_1, H_2, H_3)$ corresponds to a $Z$-logical operator and let $(S,T,U,V)$ be the choice of generators leading to the minimum weight representative of the logical error, similarly to \eqref{eqn:logical}.
The logical operator admits the form (replacing the shorthands $x,y,z$ by their explicit forms):
\begin{align*}
A = \begin{bmatrix} H_1^T S\\ H_2^T T   \\ H_3^T U+R \end{bmatrix} \quad
B = \begin{bmatrix} H_1T \\ H_2 S \\ H_3^T V+R \end{bmatrix} \quad
C = \begin{bmatrix} H_1U \\ H_2^T V \\ H_3 S \end{bmatrix} \quad
D = \begin{bmatrix} H_1^T V \\ H_2 U\\ H_3 T \end{bmatrix}?
\end{align*}
We note that a representative of the $Z$-logical operator of the new code $\mathcal{Q}(\pi_1 H_1\tau_1,\pi_2 H_2 \tau_2,\pi_3 H_3\tau_3)$ is given by 2 full horizontal layers of Pauli-$Z$ operators in $A$ and $B$ with support $\tau_3^T R$. 
Defining $S', T', U' V'$ as 
\[
S'
= (\pi_1 \otimes \tau_2^T \otimes \tau_3^T ) S, \;\;
T'
=( \tau_1^T \otimes \pi_2 \otimes \tau_3^T)T, \;\;
U'
=( \tau_1^T \otimes \tau_2^T \otimes \pi_3) U, \;\;
V'
= (\pi_1 \otimes \pi_2 \otimes \pi_3) V,
\]
we obtain the following representative of the logical error:
\begin{align*}
A' &= \begin{bmatrix} (\pi_1 H_1\tau_1)^T S'\\ (\pi_2 H_2 \tau_2)^T T'   \\ (\pi_3 H_3\tau_3)^T U'+ \tau_3^T R \end{bmatrix} = (\tau_1^T \otimes \tau_2^T \otimes \tau_3^T) \begin{bmatrix} H_1^T S\\ H_2^T T   \\ H_3^T U+R \end{bmatrix}\\
B' &= \begin{bmatrix} (\pi_1 H_1\tau_1)T' \\ (\pi_2 H_2 \tau_2) S' \\ (\pi_3 H_3\tau_3)^T V'+\tau_3^T R \end{bmatrix} =(\pi_1\otimes \pi_2 \otimes \tau_3^T) \begin{bmatrix} H_1T \\ H_2 S \\ H_3^T V+R \end{bmatrix} \\
C' &= \begin{bmatrix} (\pi_1 H_1\tau_1)U' \\ (\pi_2 H_2 \tau_2)^T V' \\ (\pi_3 H_3\tau_3) S' \end{bmatrix} = (\pi_1\otimes \tau_2^T \otimes \pi_3)\begin{bmatrix} H_1U \\ H_2^T V \\ H_3 S \end{bmatrix}\\
D' &= \begin{bmatrix} (\pi_1 H_1\tau_1)^T V' \\ (\pi_2 H_2 \tau_2) U'\\ (\pi_3 H_3\tau_3) T' \end{bmatrix}=(\tau_1^T \otimes \pi_2 \otimes \pi_3)\begin{bmatrix} H_1^T V \\ H_2 U\\ H_3 T \end{bmatrix},
\end{align*}
which has the same weight as the minimal logical error of $\mathcal{Q}(H_1,H_2,H_3)$.
\end{proof}

\subsubsection{Ruling out expansion-based arguments}
\label{sub:expansion}

It is tempting to consider the XYZ product built from classical expander codes.
Indeed, we might be able to use the expansion properties to get a good lower bound on the minimum distance.
Recall for instance that classical expander codes have a linear minimum distance \cite{SS96} and that quantum expander codes (built from classical expander codes) saturate the $\Theta(\sqrt{N})$ upper bound for the distance of hypergraph product codes \cite{LTZ15}.
Even more so, quantum expander codes have the property that any error $\cE$ of weight $O(N^{1/2})$ is expanding in the sense that its syndrome $\sigma(\cE)$ is large:
\[
\forall \cE \quad \text{s.t.} \quad |\cE| \leq \gamma \sqrt{N}, \quad
|\sigma(\cE)| \geq c |\cE|,
\]
for some constants $\gamma, c >0$.

When taking the XYZ product of three classical expander codes of size $n_i = \Theta(n)$, one could naively expect a similar phenomenon to hold for errors of weight up to $N^{2/3} = n^2$.
Unfortunately, it seems that for any $\delta>0$ we can construct errors of weight $n^{2 - \delta} = \Theta(N^{2/3-\delta/3})$ for which the syndrome has weight\footnote{A caveat is that it is not entirely clear at this point that the error cannot be reduced modulo the stabilizer, although this seems unlikely.} only $O(n)$.

The construction is algebraic in nature, and goes as follows: consider an error operator that corresponds to two horizontal slices (not necessarily full slices) of $Z$-errors in $A$ and $B$.
Denote the support of the $Z$-errors on these slices by the matrices or 2-tensors $P$ and $Q$, respectively.
Then the syndrome (in matrix notation) is given by the corresponding slices in $S$ and $T$:
\[
S
= H_1 P + Q H_2, \qquad
T
= P H_2^T + H_1^T Q.
\]
We will define the error tensors $P$ and $Q$ as follows: for some matrix $X$, let
\begin{align*}
P
&:=  H_1^T  \left( \sum_{\ell =0}^{k-1} (H_1 H_1^T)^\ell X (H_2 H_2^T)^{k-\ell-1} \right) H_2\\ 
Q
&:=  (H_1 H_1^T) \left( \sum_{\ell=0}^{k-1}  (H_1 H_1^T)^\ell X (H_2 H_2^T)^{k-\ell-1} \right).
\end{align*}
Now we calculate the syndrome associated to error $(P,Q)$:
\begin{align*}
S
&= H_1 P + Q H_2
= 0 \\
T
&= P H_2^T + H_1^T Q \\
&= H_1^T \left( \sum_{\ell=0}^{k-1} (H_1 H_1^T)^\ell X (H_2 H_2^T)^{k-\ell}
	+ \sum_{\ell=0}^{k-1} (H_1 H_1^T)^{\ell+1} X (H_2 H_2^T)^{k-\ell-1} \right) \\
&= H_1^T \left( X (H_2 H_2^T)^k + (H_1 H_1^T)^k X \right).
\end{align*}
The most important observation now is that almost all of the terms of the original error cancel, and it seems likely that this will prevent expansion of the errors.
Indeed, assume that $X$ has Hamming weight $|X| \in O(1)$.
By the expansion of $H_1$ and $H_2$, we expect that the error weight $|P| + |Q|$ will increase more or less monotonically with $k$ (up to $\Theta(n^2)$).
Hence there will be a range of $k$ such that the error has weight $|P| + |Q| \in \omega(n)$.
For the error in this range to be expanding, we need that also the syndrome has weight $|S| + |T| \in \omega(n)$.
However, as is apparent from the expressions for $S$ and $T$, it always holds that the syndrome weight $|S| + |T| \in O(n)$ (both terms in $T$ are supported on only $O(1)$ rows or columns).

This construction shows that expansion arguments are unlikely to be useful in the study of XYZ product codes: while it is easy to show that errors of weight $O(N^{1/3})$ expand (simply by applying the same method as for the quantum expander codes \cite{LTZ15}), our argument suggests that similar methods fail for errors of weight $\omega(N^{1/3})$.

\subsection{XYZ product codes from parity-check matrices with even row/column weight}

\label{subsec:decoupling-chamon}

Let us first focus on the Chamon code which provides a specific instance of XYZ product code with parity-check matrices with rows and columns of weight 2. 
Lemma \ref{lem:perm-inv} shows that the Chamon code has the same distance as $\mathcal{Q}(H_1,H_2,H_3)$ with $H_i = \Omega_{n_i} + \Omega_{n_i}^T$ when the $n_i$'s are odd (so that the map $ \ell \to 2\ell \mod n_i$ corresponds to a permutation of $[n_i]$).
Let us consider for simplicity the case where $\gcd(n_1, n_2, n_3)=1$ yielding a code of dimension 4 \cite{BLT11}. 
Any representation of the $\overline{Z}_i$ logical operator ($i\in \{1,2,3,4\}$) takes the form
\begin{align*}
A_i = \begin{bmatrix} H_1 S\\ H_2 T   \\ H_3 U+ \delta_{i,1} R \end{bmatrix} \quad
B_i = \begin{bmatrix} H_1T \\ H_2 S \\ H_3 V+ \delta_{i,2} R \end{bmatrix} \quad
C_i = \begin{bmatrix} H_1U \\ H_2 V \\ H_3 S+ \delta_{i,3} R \end{bmatrix} \quad
D_i = \begin{bmatrix} H_1 V \\ H_2 U\\ H_3 T+ \delta_{i,4} R \end{bmatrix},
\end{align*}
where $(S,T,U,V)$ is an arbitrary stabilizer element and $\delta_{i,1} = 1$ if $i=1$ and 0 otherwise. 
In particular, the logical operator $\overline{Z}_1 \overline{Z}_2 \overline{Z}_3 \overline{Z}_4$ corresponds to $(A_1, B_2, C_3, D_4)$. Its weight gives an \emph{upper} bound on the distance $d$ of the Chamon code:
\begin{align*}
d &\leq w(\overline{Z}_1 \overline{Z}_2 \overline{Z}_3 \overline{Z}_4) \\
&= \min_{ S,T,U,V} (|A_1| + |B_2| + |C_3| + |D_4|)\\
&\leq  \min_{ S = T = U =V} (|A_1| + |B_2| + |C_3| + |D_4|)\\
& = d_*
\end{align*}
where we imposed the constraint $S=T=U=V$ in the third line and defined $d_*$ as
\begin{align}\label{eqn:d*-chamon}
d_* := \min_{M} 2 \left(  |(H_1 +H_2)M | + |(H_1+H_3)M+R| + |(H_2+H_3)M + R|  \right).
\end{align}

The argument generalizes immediately to any XYZ product code where $H_i=H_i^T$, with odd $n_i$'s. 
The triangle inequality $|A|+|B|+|C|+|D| \geq |A+B+C+D|$ also shows that
\[ d_* \leq w(\overline{Z}_1) \]
by setting $M=S+T+U+V$. This does not directly lead to a lower bound on the minimum distance, however, because one needs to consider all $4^k-1$ possible logical operators of the code.

\section{Cyclic 3D XYZ code and polynomial formalism}
\label{sec:cyclic}

We focus in this section on translation-invariant codes. In the next section, we will study in greater detail one of the simplest instances of such codes, which can be embedded with local interaction in 3 dimensions.

In order to get a translation invariant code, it is possible to take parity-check matrices which are \textit{circulant matrices}.
A square matrix of dimension is called circulant when every row is shifted by one element w.r.t.~the preceding row.
Such a matrix can be described by a polynomial $H = P(\Omega)$ in the permutation matrix $\Omega$, defined by $\Omega|j\rangle = |j+1 \mod n\rangle$ (so that $\Omega^n = \1_n$).
Alternatively, these operators can be described even more concisely using a \textit{polynomial formalism} where we describe $\Omega$ using a variable $x$, and a circulant matrix $H = P(\Omega)$ is described by the polynomial $P(x) \in \F_2[x]/(1+x^n)$ (\textit{i.e.}, the polynomial ring over $\F_2$ with $x^n = 1$).

With this formalism we wish to revisit the question about the minimum distance, and in particular the decoupling argument (Section \ref{sec:decoupling}) which showed that the minimum distance problem reduces to a problem about 3-tensors.
In the polynomial formalism we describe such an $n_1 \times n_2 \times n_3$ tensor by a multivariate polynomial
\[
P(x,y,z)
\in \F_2[x,y,z]/(x^{n_1}+1, y^{n_2}+1, z^{n_3}+1)
=: \F_2[x,y,z]/I,
\]
using the shorthand $I$ for the ideal $(x^{n_1}+1, y^{n_2}+1, z^{n_3}+1)$.
The presence of a monomial $x^i y^j z^k$ in $P$ denotes the presence of a 1 in position $(i,j,k)$ of the tensor.
As an example, consider the natural $Z$-logical operator ($R$ in Section \ref{sec:decoupling}) corresponding to a full horizontal slice.
In the polynomial formalism we can describe it using the polynomial
\[
R(x,y,z)
= \sum_{i=0}^{n_1-1} \sum_{j=0}^{n_2-1} x^i y^j.
\]
If we want that the code can be embedded in 3 dimensions with local interactions, then it suffices to pick polynomials $P_i$ (describing the $H_i$'s) of the form
\begin{align*} 
P_1(x) &= 1 + \sum_{i \in \mathcal{I}} x^{i} + x^{-i} \\
P_2(y) &= 1 + \sum_{j \in \mathcal{J}} y^{j} + y^{-j}\\
P_3(z) &= 1 + \sum_{k \in \mathcal{K}} z^{k} + z^{-k} ,
\end{align*}
with $\max(i \: : \: i \in \mathcal{I}), \max(j \: : \: j \in \mathcal{J}), \max(k \: : \: k \in \mathcal{K}) = O(1)$. Note that $x^{-i}$ corresponds to $x^{n_1-i}$.
Here, we restrict to ``symmetric'' polynomials (satisfying the symmetry $x \leftrightarrow x^{-1}$) for simplicity. 
The fact that each of the sets $\mathcal{I}, \mathcal{J}, \mathcal{K}$ has constant size is equivalent to saying that the quantum code is LDPC.
Indeed, the weight of a generator will be $ 3 +2 (|\mathcal{I}| + |\mathcal{J}| + |\mathcal{K}|)$.

Now we are ready to revisit the decoupling argument.
Define the polynomial $Q_1$ by
\begin{align}\label{eqn:Qi}
Q_1(x) = (1+P_1(x))^2 :=  \sum_{i \in \mathcal{I}} x^{2i} + x^{-2i},
\end{align}
and similarly for $Q_2$ and $Q_3$.
Using these polynomials, and Theorem \ref{thm:decoupling}, we see that the minimum distance (up to constant factors, see the remark below Theorem \ref{thm:decoupling}) is given by 
\[
\min_{P \in \F_{2}[x,y,z]/I}
\Big( |(Q_1 + Q_2) P| +  |(Q_1 + Q_3) P + R|  \Big),
\]
where we omit the variables for conciseness.
This minimum is at most equal to $|R| = n_1 n_2$, which is achieved by setting $P=0$. 
The main question is whether there exists a choice for $P$ such that the distance is significantly less.
While we will not solve this question here, it is interesting to understand which properties a polynomial $P$ should display to get a small weight for $(Q_1+Q_2)P$.

\subsection{Fractal operators}

For simplicity, we would like to focus on a choice of polynomials that ensures that the dimension of the quantum code is equal to 1.
As explained in Section \ref{sec:decoupling}, taking a triple $(H_1, H_2, H_3) \in \mathcal{T}$ ensures that the map
\[
\Phi_{12}\: : \: P(x,y) \mapsto (Q_1(x) + Q_2(y)) P(x,y)
\]
is almost bijective on $F_2[x,y]/(x^{n_1}+1, y^{n_2}+1)$.
In fact, the image is equal to all polynomials with even weight and the kernel of the map has dimension 1:
\[
\ker \Phi_{12}
= \left\{0, \sum_{i=0}^{n_1-1} \sum_{j=0}^{n_2-1} x^i y^j \right\}.
\]
To understand the structure of polynomials such that $|\Phi_{12}(P)|$ is small, it is useful to recall the \textit{Froebenius endomorphism} $F$ on the polynomial ring, which is given by
\[
F(p(x,y))
= p(x^2, y^2),
\]
with the property that $F(P+Q) = F(P) + F(Q)$. 
In particular, it implies that 
\[
(Q_1(x) + Q_2(y))^{2^k}
= (Q_1(x))^{2^k} + (Q_2(y))^{2^k}
= Q_1(x^{2^k}) + Q_2(y^{2^k}),
\]
and hence
\[
\Phi_{12}\big( (Q_1(x)+Q_2(y))^{2^k-1} \big)
= (Q_1(x)+Q_2(y))^{2^k}
= Q_1(x^{2^k}) + Q_2(y^{2^k}).
\]
This gives the bound
\begin{align*}
\left| \Phi_{12}\left( (Q_1(x)+Q_2(y))^{2^k-1} \right) \right|
= \big| Q_1(x^{2^k}) + Q_2(y^{2^k}) \big|
\leq |Q_1(x)| + |Q_2(y)|.
\end{align*}
Note that we do not get an equality in general because cancellations might occur modulo $n_1$ or $n_2$. 
A polynomial of the form $(Q_1(x)+Q_2(y))^{2^k-1}$ will be referred to as a \emph{fractal operator}.

It is also interesting to understand what polynomials $P$ are such that $| (Q_1+Q_3)P + R|$ is small. 
For this, taking a fractal operator is not sufficient: while $(Q_1+Q_3)P$ might have some support on the support of $R$, it also has support on the complement of $R$. 
However, defining $m_3 := 2^{n_3-1}-1$, $P_{13}(x,z) := (Q_1(x) + Q_3(z))^{m_3}+1$, and $\Phi_{13}$ in analogy to $\Phi_{12}$, we get:
\begin{align*}
\Phi_{13}(P_{13}(x,z))
&= (Q_1(x)+Q_3(z)) P_{13}(x,z)\\
&= (Q_1(x) + Q_3(z))^{m_3+1}+Q_1(x) + Q_3(z)\\
&=  (Q_1(x) + Q_3(z))^{2^{n_3-1}}+Q_1(x) + Q_3(z)\\
&=  Q_1(x)^{2^{n_3-1}}+Q_1(x)  + Q_3(z)^{2^{n_3-1}} + Q_3(z)\\
&=  Q_1(x)^{2^{n_3-1}}+Q_1(x)  + Q_3(z^{2^{n_3-1}}) + Q_3(z)\\
&= Q_1(x)^{2^{n_3-1}}+Q_1(x) 
\end{align*}
where the last equality follows from the fact that $2^{n_3-1} = 1 \mod n_3$ and therefore $z^{2^{n_3-1}} = z$ on the ring $\F_2[z]/(z^{n_3}+1)$.
We get that $\Phi_{13}(P_{13})$ is supported on a single row (in the $x$-direction), having weight at most $n_1$.

Now in order to get a polynomial $P$ such that $| (Q_1+Q_3)P + R|$ is small, one could therefore choose a polynomial of the form
\[
P(x,y,z)
= \sum_{y=0}^{n_2-1} y^j P_j(x,z),
\]
which corresponds to $y$-translates of the $x$-$z$ planes $P_j(x,z)$.
If we choose $P_j(x,z)$ close to polynomials of the form $S_j(x) P_{13}(x,z)$, then the weight of $S_j(x) (Q_1(x)^{2^{n_3-1}}+Q_1(x) )$ is close to $n_1$ (by our former argument), which in turn implies that $| (Q_1+Q_3)P + R|$ is small.
At the same time, writing this polynomial as $P(x,y,z) = \sum_{k=0}^{n_3-1} z^k P_k(x,y)$, one asks whether each $P_k$ can be a small sum of fractal operators, to ensure that $|(Q_1+Q_2)P_k|$ is small.  It is tempting here to try to argue that if $P$ was chosen so that $|(Q_1+Q_3)P + R|$ is small, then there must exist some choice of $Q_2$ for which the $|(Q_1+Q_2)P_k|$ cannot all be small. 
Formalizing a probabilistic argument along these lines may be an approach to proving the existence of XYZ product codes with distance $\Theta(N^{2/3})$ but appears to be a complicated problem. 

In the next section, we will argue that taking $Q_1=Q_2 = Q_3$, but distinct $n_i$'s (which is for instance the case of the original Chamon's code) may often lead to the existence of low-weight logical operators, and therefore a distance possibly as small as $\Theta(N^{1/3})$.

We remark that the existence of fractal operators which can have a linear weight but a constant-weight syndrome shows that the cyclic XYZ product codes investigated here cannot be locally testable \cite{AE15}.

\subsection{Embedding in 3-dimensional space}
\label{sub:embed}

Let us quickly remark on the possibility to locally embed such a code in three dimensions.
We consider a lattice of size $n_1 \times n_2 \times n_3$ with boundary conditions. Putting 4 qubits per site (one qubit of each type, $A, B,C,D$), one obtains a topological code with generators of weight $d_1+d_2+d_3$, where $d_i$ is the weight of $P_i$.
To ensure that the quantum code admits geometrically local generators, it is sufficient that the three classical codes are local in 1 dimension.


\section{The simplest 3D XYZ code}
\label{sec:simplest}

A general upper bound on the distance for 3-dimensional topological codes is $d = O(N^{2/3})$ \cite{BT09} and it is an open question whether this bound is tight or whether the correct value is $O(N^{1/2})$ as suggested by all known code constructions.
We try to shed new light on this question by investigating one of the simplest XYZ product codes embeddable in 3 dimensions.
Let ${\cal Q}(H_1,H_2,H_3)$ be the XYZ product code derived from the parity-check matrices
\[
H_i
= \1_{n_i} + \Omega_{n_i} + \Omega_{n_i}^T,
\]
with $\Omega_{n_i} |j\rangle = | j+1 \mod n_i\rangle$.
The only freedom in the code construction is the choice of the triple $(n_1, n_2, n_3)$.
We will first compute the dimension of this 3D XYZ code when the $n_i$'s are odd and not multiples of 3.
The formula for the dimension is similar to that of the Chamon code, but with what looks like an extra qubit.
By taking the $n_i$'s to be coprime, we can ensure that only this logical qubit is encoded, and so the code has dimension 1.
For this particular case we show that most choices of $n_i$'s lead to a distance upper bounded by $N^{1/3}$. While this does not rule out the existence of families of $n_i$'s leading to a larger distance, it strongly suggests that taking slightly more complicated parity-check matrices will be required to beat the $\sqrt{N}$ barrier.
Indeed, recall from Lemma \ref{lem:perm-inv} that permuting the rows and columns of the parity-check matrices does not change the minimum distance.
As a consequence, any circulant matrices with columns of weight 3 will lead to the same distance, and columns of weight 5 will probably be necessary to get better parameters.

\begin{rem}
Ref.~\cite{BLT11}, which focused on the Chamon code, also suggested some generalizations where the repetition code would be replaced by symmetric circulant matrices with even row weight.
Unfortunately, the decoupling argument of Theorem \ref{thm:decoupling} does not apply in that case since the 3-tensor $R$ corresponding to a single horizontal plane of 1s is not stabilized by circulant matrices of even weight. 
\end{rem}

\subsection{Dimension of the 3D XYZ code}
\label{sub:dim-3D}

Using Theorem \ref{thm:xyz-dim} we can prove the following claim on the dimension of the 3D XYZ code.
\begin{claim} \label{claim:dim-cubic}
If $n_1$, $n_2$ and $n_3$ are odd and no multiples of 3, then the dimension of the 3D XYZ code is
\[
4(\gcd(n_1,n_2,n_3)-1) + 1.
\]
\end{claim}

To prove the claim, first fix some $n$ and let $H = \1 + \Omega + \Omega^T$ be the corresponding matrix of dimension $n$.
We can prove the following lemma.
\begin{lemma}
Assume $n$ is odd and let $z$ be a primitive $n$-th root of unity\footnote{A solution of $z^n=1$ such that $z^k\neq 1$ for $k < n$. Such a root always exists when $n$ is odd.}.
Then the matrix $H = \1 + \Omega + \Omega^T$ is diagonalizable (over the algebraic closure of $\overline{\mathbb{F}}_2$) with eigenvalues
\[
1, 1 + z + z^{-1}, \dots, 1 + z^{n-1} + z^{-n+1}.
\]
It is invertible if and only if $n$ is not a multiple of 3.
\end{lemma}
\begin{proof}
First we note that the characteristic polynomial of $\Omega$ is $x^n + 1$, so that its eigenvalues are given by the $n$-th roots of unity.
Over any field of characteristic 2, if $n$ is odd then there exist $n$ distinct roots of unity $1,z,z^2,\dots,z^{n-1}$.
This implies that $\Omega$ is diagonalizable with eigenvalues $1,z,\dots,z^{n-1}$.
The first claim then follows by noting that $\Omega^T = \Omega^{-1}$, and that $\Omega^{-1}$ is diagonal in the same basis with eigenvalues $1,z^{-1},\dots,z^{-n+1}$.

$H$ will be invertible if there is no $0 \leq k \leq n-1$ such that $1 + z^k + z^{n-k} = 0$.
This is equivalent to $1 + z^k + z^{2k} = 0$, and combining these equations shows that necessarily $z^{2k} = z^{n-k}$, or equivalently $z^{3k} = z^n = 1$.
Since $z$ is a primitive root, this implies that $k = n/3$ or $k = 2n/3$, so that there is a solution if and only if $n$ is a multiple of 3.
\end{proof}

Now assume that the $n_i$'s are odd and no multiples of 3, ensuring that the $H_i$'s are invertible.
Moreover, $H_i^2$ and $H_i$ are similar if $n$ is odd.
From Theorem \ref{thm:xyz-dim} it follows that the dimension of the code is $\sum_{i,j,k} \mathrm{deg}(\mathrm{gcd}(p^1_i,p^2_j,p^3_k))$ with $p^1_i,p^2_j,p^3_k$ the characteristic polynomials of the Jordan blocks of $H_1^2,H_2^2,H_3^2$ (respectively).
By the lemma above the Jordan blocks are 1-dimensional and hence the dimension is simply
\[
\sum_{i,j,k} \mathrm{deg}(\mathrm{gcd}(p^1_i,p^2_j,p^3_k))
= \sum_{i=0}^{n_1-1} \sum_{j=0}^{n_2-1} \sum_{k=0}^{n_3-1}
	\mathbf{1}(\lambda_{n_1}^i=\lambda_{n_2}^j=\lambda_{n_3}^k),
\]
where $\lambda_{n_i}^j = 1 + z_{n_i}^j + z_{n_i}^{-j}$ with $z_{n_i}$ a primitive $n_i$-th root of unity.
This corresponds to the number of solutions to the equation
\[
x^a + x^{-a} = y^b + y^{-b} = z^c + z^{-c}
\]
for $0 \leq a < n_1$, $0 \leq b < n_2$, $0 \leq c < n_3$ and $x,y,z$ primitive $n_1$-, $n_2$- and $n_3$-th roots of unity.
By the following lemma this proves Claim \ref{claim:dim-cubic}.
\begin{lemma}
Let $x,y,z$ be primitive $k$-, $\ell$-, and $m$-th roots of unity (respectively).
Then the number of solutions to the equation $x^a + x^{-a} = y^b + y^{-b} = z^c + z^{-c}$ (with $0 \leq a < k$, $0\leq b < \ell$ and $0 \leq c < m$) is equal to $4(\gcd(k,\ell,m)-1)+1$ 
\end{lemma}
\begin{proof}
First we show that the number of solutions is $4(q-1)+1$ with $q$ the number of solutions to the equation $x^a = y^b = z^c$ (with $0 \leq a < k$, $0\leq b < l$ and $0 \leq c < m$).
To this end we use that $x + 1/x = y + 1/y$ (for $x,y \neq 0$) implies that $x = y$ or $x = 1/y$.
To see this, note that $x + 1/x = y + 1/y$ is equivalent to $x^2 + (y+1/y)x + 1 = (x+y)(x+1/y) = 0$.
On its turn, this implies that $x^a + x^{-a} = y^b + y^{-b} = z^c + z^{-c}$ is equivalent to $x^a = y^{\pm b} = z^{\pm c}$ (for some choice of the signs).
The number of solutions to $x^a = y^{\pm b} = z^{\pm c}$ is then given by $4(q-1) + 1$ where $q$ is the number of solutions to $x^a = y^b = z^c$.
Indeed, every solution $(a,b,c)$ of $x^a = y^b = z^c$ with $a,b,c \neq 0$ gives rise to 4 distinct solutions $(a,\pm b, \pm c)$ of $x^a = y^{\pm b} = z^{\pm c}$ (using that $y^b \neq y^{-b},z^c \neq z^{-c}$ for $0 < b < \ell$ and $0 < c < m$).
On the other hand, if either $a=0$, $b=0$ or $c=0$ then necessarily $a=b=c=0$ (since the roots are primitive), so this case corresponds to a single solution of $x^a = y^b = z^c$, and it gives rise to only 1 solution to $x^a = y^{\pm b} = z^{\pm c}$.

Next we show that the number of solutions $(a,b,c)$ to $x^a = y^b = z^c$ is $q = \gcd(k,\ell,m)$.
First note that if $(a,b,c)$ is a solution then $x^{a k} = x^{a \ell} = x^{a m} = 1$.
Hence $a\ell = qk$ and $am = rk$ for integers $q$ and $r$.
This implies that $am\ell = mqk = rk\ell$ is a common multiple of $k\ell$, $\ell m$ and $km$.
Moreover every $a < k$ defines a unique such common multiple of size $<k\ell m$ (as well as unique choices for $b$ and $c$), and every common multiple $s$ defines a solution by setting $(a,b,c)=(s/\ell m,s/km,s/k\ell)$.
Now we can use the formula
\[
\mathrm{lcm}(k\ell,\ell m,km)
= \frac{k\ell m}{\gcd(k,\ell,m)}.
\]
Since every common multiple of $k,\ell,m$ is a multiple of $\mathrm{lcm}(k\ell,\ell m,km)$, this shows that there are exactly $\gcd(k,\ell,m)-1$ nontrivial (different from 0) common multiples of size $<k\ell m$.
Adding the trivial solution $a=b=c=0$ finishes the lemma.
\end{proof}

\subsection{Minimum distance of the 3D XYZ code}
\label{sub:dmin-3D}

In this section, we choose the $n_i$'s so that the 3D XYZ code has dimension 1. More precisely, we choose $n_1, n_2, n_3$ to be coprime, odd, and not divisible by 3. 
In this case, there are only three logical operators to consider.
By the symmetry of the code, we can focus on a $Z$-logical operator without loss of generality (but by considering all the permutations of $n_1, n_2, n_3$).
Indeed the weight a $Z$-logical operator for $(n_1, n_2, n_3)$ is identical to the weight of an $X$-logical operator of the code parameterized by $(n_3,n_2,n_1)$.

We adopt the same polynomial notations as in the previous section, in which an arbitrary 3-tensor is described by a polynomial $P$ in the ring $\F_2[x,y,z]/\I$ with the ideal $\I := (x^{n_1} + 1, y^{n_2}+1, z^{n_3}+1)$.
Up to some multiplicative constant, Theorem \ref{thm:decoupling} (see the remark below this theorem) shows that the minimum distance is given by 
\begin{align}\label{eqn:dmin3D}
\min_P | (1+xy)(1+x/y)P | + |(1+xz)(1+x/z)P + R|,
\end{align}
where $R = \sum_{i=0}^{n_1-1}\sum_{j=0}^{n_2-1} x^i y^j$ corresponds to a single horizontal plane in the 3D picture. 
To see this, recall that the minimum distance is unchanged if we perform the change of variables $x \to x^{(n_1+1)/2}, y \to y^{(n_2+1)/2}, z \to z^{(n_3+1)/2}$ which yields the following polynomials $Q_1, Q_2, Q_3$ (defined in Eqn.~\eqref{eqn:Qi}):
\[ Q_1(x) = x^{-1}+x, \quad Q_2(y) = y^{-1}+y, \quad Q_3(z) = z^{-1}+z.\]
Note moreover that 
\[ Q_1(x) + Q_2(y) = x^{-1}(1+xy)(1+x/y) \quad \text{and} \quad Q_1(x) + Q_3(z) = x^{-1}(1+xz)(1+x/z).\]
Recall that the weight of a polynomial (that is the number of monomials) is invariant when the polynomial is multiplied by a (nonzero) monomial, corresponding to a translation in the tensor picture.

The main strategy to form a small weight solution of \eqref{eqn:dmin3D} is to consider $P$ of the form 
\[ P(x,y,z) = \left( \sum_{j=0}^{n_2-1} (xy)^j\right) P_{13}(x,z),\]
with $P_{13}$ such that 
\[  \left|(1+xz)(1+x/z)P_{13}(x,z) + \sum_{i=0}^{n_1-1} x^i \right| = f(n),\]
for some function $f(n) \in o(n)$. 
This choice yields $ |(1+xz)(1+x/z)P + R| = n f(n)$ and 
\begin{align*}
 (1+xy)(1+x/y)P   &=  (1+xy)  \left(\sum_{j=0}^{n_2-1} (xy)^j\right) (1+x/y)P_{13}(x,z)  \\
&=  (1+(xy)^{n_2})  (1+x/y)P_{13}(x,z)  \\
&=   (1+x^{n_2})  (1+x/y)P_{13}(x,z) 
\end{align*}
where we used $y^{n_2}=1$ in the last equality.
In particular, we obtain 
\[  |(1+xy)(1+x/y)P | \leq 2 | (1+x^{n_2})  P_{13}(x,z) |.\]

While it is possible to choose a triple $(n_1, n_2, n_3)$ such that the right-hand side is large for any $P_{13}$ achieving a low value of $f(n)$, we have not been able to find any triple such that $ | (1+x^{n_j})  P_{ik}|$ remains large for all $P_{ik}$ achieving a small $f(n)$, for all permutations of $(i,j,k)$. 
We are therefore tempted to conjecture that this simple XYZ product code has a minimum distance $d= o(N^{2/3})$ for any choice of $n_i$'s. In fact, any cyclic XYZ product code with $P_1 = P_2 = P_3$ such that $P_i$ is local in 1 dimension (and therefore the XYZ code is local in 3 dimensions) will allow similar constructions. This means that beating the $\sqrt{N}$ barrier for the minimum distance will likely require to take 3 distinct polynomials. 

We detail in the next section the simplified case of the Chamon code.

\subsection{$O(\sqrt{N})$ upper bound on the distance of the Chamon code}

While the analysis in this section does not improve on the upper bound from \cite{BLT11}, it is still interesting because it generalizes to cyclic XYZ product codes with polynomials with larger weight. Moreover, it shows that a $\sqrt{N}$ distance for the Chamon code does not correspond to the generic case of randomly chosen coprime $n_1, n_2, n_3$. On the contrary, it seems that for most triples, the distance is more likely to be $\Theta(N^{1/3})$, which raises the question of the existence of triples yielding a better distance.

For the Chamon code, the polynomials $Q_i$ take a very simple form
\[ Q_1(x) = x, \quad Q_2(y)=y, \quad Q_3(z)=z,\]
and Section \ref{subsec:decoupling-chamon} shows that the minimum distance is therefore \emph{upper bounded} by 
\[ \min_{P} |(1+xy^{-1}) P| + |(1+xz^{-1})P+R|\]
up to a multiplicative constant. As before, the minimization is over polynomials $P \in \F_2[x,y,z]/\I$ with the ideal $\I := (x^{n_1} + 1, y^{n_2}+1, z^{n_3}+1)$.

We consider the case where $\gcd(n_1, n_2,n_3)=1$ and set $m_1$ such that 
\begin{align}\label{eqn:m1}
m_1 n_2 = n_3 \mod n_1.
\end{align}
Define $P(x,y,z)$ as follows:
\[ P(x,y,z) = \left(\sum_{i=0}^{m_1-1} x^{i n_2}  \right)  \left(\sum_{j=0}^{n_2-1} (x y^{-1})^j  \right)  \left(\sum_{k=0}^{n_3-1} (xz^{-1})^k  \right).\] 
We first observe that 
\[ (1+xy^{-1})P =  \left(\sum_{i=0}^{m_1-1} x^{i n_2}  \right) (1+x^{n_2})   \left(\sum_{k=0}^{n_3-1} (xz^{-1})^k  \right) \]
where we performed the simplification $1+(xy^{-1})^{n_2} = 1+x^{n_2}$ over the ring. 
This further simplifies into $ (1+x^{m_1 n_2}) \left(\sum_{k=0}^{n_3-1} (xz^{-1})^k  \right)$ 
with weight
\[ |(1+xy^{-1})P| \leq 2 n_3.\]
On the other hand, 
\begin{align*}
(1+xz^{-1})P+R &= R +  \left(\sum_{i=0}^{m_1-1} x^{i n_2}  \right)  \left(\sum_{j=0}^{n_2-1} (x y^{-1})^j  \right)  \left(1+x^{n_3}  \right).
\end{align*}
The definition of $m_1$ ensures that 
\[  \left(1+x^{n_3}  \right)\left(\sum_{i=0}^{m_1-1} x^{i n_2}  \right) =\sum_{i=0}^{2m_1-1} x^{i n_2}, \]
which has weight 
\begin{align*}
w_1 := \left\{ \begin{array}{ll}
2m_1 & \text{if} \quad 2m_1 \leq n_1\\
2(n_1-m_1) & \text{if} \quad 2m_1 > n_1.
\end{array}\right.
\end{align*}
In particular, we obtain that
\begin{align*}
 |(1+xz^{-1})P+R| &= n_2 \left|(1+xz^{-1})P+ \sum_{i=0}^{n_1-1} x^i \right|\\
&= n_2 (n_1 - w_1).
\end{align*}
Our specific choice of $P$ removed for each $xz$-slice $w_1$ out of the $n_1$ coordinates in $R$, at a ``cost'' of $2n_3$ for $|(Q_1+Q_3)P|$.
This can be repeated up to $q_1 := \left\lfloor \frac{n_1}{w_1} \right \rfloor$ times by considering the polynomial 
\[ P' = \left(\sum_{i=0}^{q_1-1} x^{2 m_1 n_2 i } \right) P\]
which yields a low-weight logical operator since
\[  |(1+xy^{-1}) P'| + |(1+xz^{-1})P'+R| \leq 2 q_1 n_3 + n_2 (n_1 - q_1 w_1).\]
Note that if $w_1$ is very small, then $q_1$ will be large and this will not provide an interesting upper bound on the distance. 
Remember, however, that one needs to consider all permutations of indices. In particular, from the definition of $m_1$ in Eqn.~\eqref{eqn:m1}, we infer that $m_1'$ defined as
\[ m_1' n_3 = n_2 \mod n_1\]
obtained by exchanging the roles of $n_2$ and $n_3$ satisfies
\[m_1 m_1' = 1 \mod n_1.\]
In particular, either $m_1$ (and $n_1-m_1$) or $m_1'$ (and $n_1-m_1'$) is $\Omega(\sqrt{n_1})$. We therefore assume without loss of generality that $w_1 = \Omega(\sqrt{n_1})$ and keep the same setting as before.
From this, we get that 
\[ 2 q_1 n_3 = O (\sqrt{n_1} n_3) = O(N^{1/2})\]
in the regime where $n_1, n_2, n_3 = O(N^{1/3})$.

In the slice $j=0$ (\textit{i.e.}, corresponding to constant term in $y$), the remaining coordinates in the support of $R$ but not covered by $(1+xz^{-1}) P$ form an arithmetic progression over $\mathbbm{Z}/n_1 \mathbbm{Z}$ with common difference $n_2$ and length $L:=n_1-q_1w_1$.
If this length is $L = O(\sqrt{n_1})$, then we are done.
Otherwise we want to show how to take an additional $P''$ that will cover most of this remaining set. 

Let us take some $r \leq  L$ and some integer $s$ to be determined later, and define $P''(x,y,z)$ as
\[ P''(x,y,z) = x^{2m_1 q_1 n_2} \left(\sum_{i=0}^{r-1} x^{i n_2}  \right)    \left(\sum_{i=0}^{s-1} x^{in_3}  \right) \left(\sum_{j=0}^{n_2-1} (x y^{-1})^j  \right)  \left(\sum_{k=0}^{n_3-1} (xz^{-1})^k  \right).\] 
It satisfies $|(1+xy^{-1})P'' | \leq 2sn_3$ with an additional factor $s$ compared to previously.
Moreover,
\begin{align*}
(1+xz^{-1})P'' &= x^{2m_1 q_1 n_2}  \left(\sum_{i=0}^{r-1} x^{i n_2}  \right)  \left(\sum_{j=0}^{n_2-1} (x y^{-1})^j  \right)  (1+x^{s n_3}).
\end{align*}
The restriction to the slice $j=0$ (\textit{i.e.}, corresponding to constant term in $y$) reads
\[ x^{2m_1 q_1 n_2}  \left(\sum_{i=0}^{r-1} x^{i n_2}  \right)  (1+x^{sn_3}).\]
We want the terms $(1+x^{sn_3}) \sum_{i=0}^{r-1} x^{i n_2}  $ to cover all the remaining arithmetic progression, except possibly for $O(\sqrt{n_3})$ terms (since recall that we assume $n_i = \Theta (N^{1/3})$). 
For this, it is sufficient to take $r= L/2 + O(\sqrt{n_3})$ as well as $s$ such that 
\[ sn_3 = rn_2 \mod n_1.\]
To establish the $O(\sqrt{N})$ upper bound on the distance of the Chamon code, it is then sufficient to show that there exists a choice of $s = O(\sqrt{n_3})$ satisfying this equation. 
Multiplying both sides by the inverse of $n_2$ modulo $n_1$ gives
\[ s n_3 n_2^{-1} = r \mod n_1,\]
that is 
\[ s m_1 = r \mod n_1.\]
Since we have $m_1, n_1-m_1 = \Omega(\sqrt{n_1})$, we infer the existence of a couple $(r,s)$ with the required properties, if $c$ is chosen large enough. 
This completes the proof that the Chamon code has distance 
\[ d_{\text{Chamon}} = O(\sqrt{N}).\]

We note that although we only prove a $O(\sqrt{N})$ upper bound, the kind of construction above suggests that the distance will likely be much lower (as low as $N^{1/3}$) for most choices of $n_i$'s and it is not entirely clear that there will exist triples of $n_i$'s saturating the upper bound.

\subsection{Lack of energy barrier of the 3D XYZ code}
\label{sub:nobarrier-3D}

Similarly to Chamon's code, we will show that the 3D XYZ code obtained by taking the product of $H_i = \1_{n_i} + \Omega_{n_i} + \Omega_{n_i}^T$ does not display an energy barrier. 
For this it suffices to exhibit a path of errors going from the identity operators to a logical operator such that, at each step, the syndrome weight is upper bounded by a constant.
To this end we will describe a general logical error, distributed over the qubits in $A,B,C,D$, by a polynomial in 12 variables $x_Q, y_Q, z_Q$ for $Q \in \{A,B,C,D\}$.
Now recall that one possible representative of a $Z$-logical operator corresponds to a full horizontal slice of $Z$-Pauli errors applied to all qubits of types $A$ and $B$, which corresponds to the polynomial
\[
E_L
= \sum_{i=0}^{n_1-1}\sum_{j=0}^{n_2-1} x_A^i y_A^j + x_B^i y_B^j.
\]
Assuming that $n_1$ and $n_2$ are coprime, the idea is to flip qubits two at a time.
At time $k$ from 0 to $n_1 n_2-1$, one flips qubits in $A$ and $B$ at location $(i_k, j_k)$ with 
\[ i_k = k \mod n_1, \quad j_k = k \mod n_2.\]
At step $k$, the error is described by the polynomial
\[
E_L(k)
= \sum_{\ell = 0}^{k-1} (x_Ay_A)^\ell + (x_B y_B)^{\ell},
\]
where we recall that we work over the ring $F_{2}[x_A,y_A,x_B,y_B]/I$ with the ideal 
\[
I
:= (x_A^{n_1}+1, y_A^{n_2} +1, x_B^{n_1}+1, y_B^{n_2} +1).
\]
Note that since $n_1$ and $n_2$ are coprime, it is possible to perform a change of variable with $w = xy$ such that 
\[
F_2[x,y]/(x^{n_1}+1, y^{n_2} +1)
\cong \F_2[w]/(w^{n_1n_2}+1).
\]
With this new notation, we can rewrite
\[
E_L(k)
= \sum_{\ell = 0}^{k-1} w_A^\ell + w_B^{\ell},
\]
from which we get that this path connects the identity to a $Z$-logical operator:
\[
E_L(0) = 0, \quad
E_L(n_1 n_2) = E_L.
\]
It is then straightforward to observe that the syndrome of the error at step $k$ lies in the neighborhood of the end points of the string error. In particular, it has constant weight. 
A similar argument applies more generally to any cyclic XYZ product code described in Section \ref{sec:cyclic} as soon as two of the three polynomials are identical, even if the $n_i$'s differ.


%

\begin{thebibliography}{42}
\providecommand{\natexlab}[1]{#1}
\providecommand{\url}[1]{\texttt{#1}}
\expandafter\ifx\csname urlstyle\endcsname\relax
  \providecommand{\doi}[1]{doi: #1}\else
  \providecommand{\doi}{doi: \begingroup \urlstyle{rm}\Url}\fi

\bibitem[Aharonov and Eldar(2015)]{AE15}
Dorit Aharonov and Lior Eldar.
\newblock Quantum locally testable codes.
\newblock \emph{SIAM Journal on Computing}, 44\penalty0 (5):\penalty0
  1230--1262, 2015.
\newblock \doi{10.1137/140975498}.

\bibitem[Ataides et~al.(2021)Ataides, Tuckett, Bartlett, Flammia, and
  Brown]{BTB20}
J~Pablo~Bonilla Ataides, David~K Tuckett, Stephen~D Bartlett, Steven~T Flammia,
  and Benjamin~J Brown.
\newblock {The XZZX surface code}.
\newblock \emph{Nature Communications}, 12\penalty0 (1):\penalty0 1--12, 2021.
\newblock \doi{10.1038/s41467-021-22274-1}.

\bibitem[Audoux and Couvreur(2019)]{AC19}
Benjamin Audoux and Alain Couvreur.
\newblock {On tensor products of CSS Codes}.
\newblock \emph{Annales de l’Institut Henri Poincar{\'e} (D) Combinatorics,
  Physics and their Interactions}, 6\penalty0 (2):\penalty0 239--287, 2019.
\newblock \doi{10.4171/AIHPD/71}.

\bibitem[Bombin and Martin-Delgado(2007)]{BM07}
H.~Bombin and M.~A. Martin-Delgado.
\newblock Exact topological quantum order in $d=3$ and beyond: Branyons and
  brane-net condensates.
\newblock \emph{Phys. Rev. B}, 75:\penalty0 075103, Feb 2007.
\newblock \doi{10.1103/PhysRevB.75.075103}.

\bibitem[Bravyi and Hastings(2014)]{BH14}
Sergey Bravyi and Matthew~B Hastings.
\newblock Homological product codes.
\newblock In \emph{Proceedings of the forty-sixth annual ACM symposium on
  Theory of computing}, pages 273--282. ACM, 2014.
\newblock \doi{10.1145/2591796.2591870}.

\bibitem[Bravyi and Terhal(2009)]{BT09}
Sergey Bravyi and Barbara Terhal.
\newblock A no-go theorem for a two-dimensional self-correcting quantum memory
  based on stabilizer codes.
\newblock \emph{New Journal of Physics}, 11\penalty0 (4):\penalty0 043029,
  2009.
\newblock \doi{10.1088/1367-2630/11/4/043029}.

\bibitem[Bravyi et~al.(2010)Bravyi, Terhal, and Leemhuis]{BTL10}
Sergey Bravyi, Barbara~M Terhal, and Bernhard Leemhuis.
\newblock Majorana fermion codes.
\newblock \emph{New Journal of Physics}, 12\penalty0 (8):\penalty0 083039,
  2010.
\newblock \doi{10.1088/1367-2630/12/8/083039}.

\bibitem[Bravyi et~al.(2011)Bravyi, Leemhuis, and Terhal]{BLT11}
Sergey Bravyi, Bernhard Leemhuis, and Barbara~M Terhal.
\newblock {Topological order in an exactly solvable 3D spin model}.
\newblock \emph{Annals of Physics}, 326\penalty0 (4):\penalty0 839--866, 2011.
\newblock \doi{10.1016/j.aop.2010.11.002}.

\bibitem[Breuckmann and Eberhardt(2021{\natexlab{a}})]{BE20}
Nikolas~P. Breuckmann and Jens~N. Eberhardt.
\newblock Balanced product quantum codes.
\newblock \emph{IEEE Transactions on Information Theory}, 67\penalty0
  (10):\penalty0 6653--6674, 2021{\natexlab{a}}.
\newblock \doi{10.1109/TIT.2021.3097347}.

\bibitem[Breuckmann and Eberhardt(2021{\natexlab{b}})]{BE21}
Nikolas~P. Breuckmann and Jens~Niklas Eberhardt.
\newblock Quantum low-density parity-check codes.
\newblock \emph{PRX Quantum}, 2:\penalty0 040101, Oct 2021{\natexlab{b}}.
\newblock \doi{10.1103/PRXQuantum.2.040101}.

\bibitem[Brown et~al.(2016)Brown, Loss, Pachos, Self, and Wootton]{BLP16}
Benjamin~J Brown, Daniel Loss, Jiannis~K Pachos, Chris~N Self, and James~R
  Wootton.
\newblock Quantum memories at finite temperature.
\newblock \emph{Reviews of Modern Physics}, 88\penalty0 (4):\penalty0 045005,
  2016.
\newblock \doi{10.1103/RevModPhys.88.045005}.

\bibitem[Calderbank and Shor(1996)]{CS96}
A~Robert Calderbank and Peter~W Shor.
\newblock Good quantum error-correcting codes exist.
\newblock \emph{Physical Review A}, 54\penalty0 (2):\penalty0 1098, 1996.
\newblock \doi{10.1103/PhysRevA.54.1098}.

\bibitem[Campbell(2019)]{C19}
Earl~T Campbell.
\newblock A theory of single-shot error correction for adversarial noise.
\newblock \emph{Quantum Science and Technology}, 4\penalty0 (2):\penalty0
  025006, feb 2019.
\newblock \doi{10.1088/2058-9565/aafc8f}.

\bibitem[Chamon(2005)]{cha05}
Claudio Chamon.
\newblock Quantum glassiness in strongly correlated clean systems: an example
  of topological overprotection.
\newblock \emph{Physical Review Letters}, 94\penalty0 (4):\penalty0 040402,
  2005.
\newblock \doi{10.1103/PhysRevLett.94.040402}.

\bibitem[Evra et~al.(2020)Evra, Kaufman, and Z{\'{e}}mor]{EKZ20}
Shai Evra, Tali Kaufman, and Gilles Z{\'{e}}mor.
\newblock Decodable quantum {LDPC} codes beyond the square root distance
  barrier using high dimensional expanders.
\newblock In \emph{61st {IEEE} Annual Symposium on Foundations of Computer
  Science, {FOCS} 2020, Durham, NC, USA, November 16-19, 2020}, pages 218--227.
  {IEEE}, 2020.
\newblock \doi{10.1109/FOCS46700.2020.00029}.

\bibitem[Fawzi et~al.(2018)Fawzi, Grospellier, and Leverrier]{FGL18}
Omar Fawzi, Antoine Grospellier, and Anthony Leverrier.
\newblock Constant overhead quantum fault-tolerance with quantum expander
  codes.
\newblock In Mikkel Thorup, editor, \emph{59th {IEEE} Annual Symposium on
  Foundations of Computer Science, {FOCS} 2018, Paris, France, October 7-9,
  2018}, pages 743--754. {IEEE} Computer Society, 2018.
\newblock \doi{10.1109/FOCS.2018.00076}.

\bibitem[Freedman and Hastings(2014)]{FH14}
Michael~H Freedman and Matthew~B Hastings.
\newblock Quantum systems on non-$k$-hyperfinite complexes: a generalization of
  classical statistical mechanics on expander graphs.
\newblock \emph{Quantum Information \& Computation}, 14\penalty0
  (1-2):\penalty0 144--180, 2014.

\bibitem[Freedman et~al.(2002)Freedman, Meyer, and Luo]{FML02}
Michael~H Freedman, David~A Meyer, and Feng Luo.
\newblock Z2-systolic freedom and quantum codes.
\newblock \emph{Mathematics of quantum computation, Chapman \& Hall/CRC}, pages
  287--320, 2002.

\bibitem[Gantmacher(1959)]{gantmacher1959theory}
Feliks~Ruvimovich Gantmacher.
\newblock \emph{The theory of matrices}, volume 131.
\newblock American Mathematical Society, 1959.

\bibitem[Gottesman(1997)]{got97}
Daniel Gottesman.
\newblock \emph{Stabilizer codes and quantum error correction}.
\newblock PhD thesis, California Institute of Technology, 1997.

\bibitem[Haah(2011)]{haa11}
Jeongwan Haah.
\newblock Local stabilizer codes in three dimensions without string logical
  operators.
\newblock \emph{Physical Review A}, 83\penalty0 (4):\penalty0 042330, 2011.
\newblock \doi{10.1103/PhysRevA.83.042330}.

\bibitem[Hastings(2017)]{has17b}
Mathew~B Hastings.
\newblock Weight reduction for quantum codes.
\newblock \emph{Quantum Information \& Computation}, 17\penalty0
  (15-16):\penalty0 1307--1334, 2017.

\bibitem[Hastings et~al.(2021)Hastings, Haah, and O'Donnell]{HHO20}
Matthew~B. Hastings, Jeongwan Haah, and Ryan O'Donnell.
\newblock Fiber bundle codes: breaking the $n^{1/2} polylog(n)$ barrier for
  quantum {LDPC} codes.
\newblock In Samir Khuller and Virginia~Vassilevska Williams, editors,
  \emph{{STOC} '21: 53rd Annual {ACM} {SIGACT} Symposium on Theory of
  Computing, Virtual Event, Italy, June 21-25, 2021}, pages 1276--1288. {ACM},
  2021.
\newblock \doi{10.1145/3406325.3451005}.
\newblock URL \url{https://doi.org/10.1145/3406325.3451005}.

\bibitem[Kaufman and Tessler(2021)]{KT20}
Tali Kaufman and Ran~J. Tessler.
\newblock {New cosystolic expanders from tensors imply explicit Quantum {LDPC}
  codes with $\Omega(n \log^k n)$ distance}.
\newblock In Samir Khuller and Virginia~Vassilevska Williams, editors,
  \emph{{STOC} '21: 53rd Annual {ACM} {SIGACT} Symposium on Theory of
  Computing, Virtual Event, Italy, June 21-25, 2021}, pages 1317--1329. {ACM},
  2021.
\newblock \doi{10.1145/3406325.3451029}.
\newblock URL \url{https://doi.org/10.1145/3406325.3451029}.

\bibitem[Kaufman et~al.(2016)Kaufman, Kazhdan, and Lubotzky]{KKL16}
Tali Kaufman, David Kazhdan, and Alexander Lubotzky.
\newblock Isoperimetric inequalities for {Ramanujan} complexes and topological
  expanders.
\newblock \emph{Geometric and Functional Analysis}, 26\penalty0 (1):\penalty0
  250--287, 2016.
\newblock \doi{10.1007/s00039-016-0362-y}.

\bibitem[Kitaev(2003)]{kit03}
A~Yu Kitaev.
\newblock Fault-tolerant quantum computation by anyons.
\newblock \emph{Annals of Physics}, 303\penalty0 (1):\penalty0 2--30, 2003.
\newblock \doi{10.1016/S0003-4916(02)00018-0}.

\bibitem[Kovalev and Pryadko(2013)]{KP13}
Alexey~A Kovalev and Leonid~P Pryadko.
\newblock Quantum {Kronecker} sum-product low-density parity-check codes with
  finite rate.
\newblock \emph{Physical Review A}, 88\penalty0 (1):\penalty0 012311, 2013.
\newblock \doi{10.1103/PhysRevA.88.012311}.

\bibitem[Leverrier et~al.(2015)Leverrier, Tillich, and Z{\'{e}}mor]{LTZ15}
Anthony Leverrier, Jean{-}Pierre Tillich, and Gilles Z{\'{e}}mor.
\newblock Quantum expander codes.
\newblock In Venkatesan Guruswami, editor, \emph{{IEEE} 56th Annual Symposium
  on Foundations of Computer Science, {FOCS} 2015, Berkeley, CA, USA, 17-20
  October, 2015}, pages 810--824. {IEEE} Computer Society, 2015.
\newblock \doi{10.1109/FOCS.2015.55}.
\newblock URL \url{https://doi.org/10.1109/FOCS.2015.55}.

\bibitem[Maurice(2014)]{mau14}
Denise Maurice.
\newblock \emph{Codes correcteurs quantiques pouvant se d{\'e}coder
  it{\'e}rativement}.
\newblock PhD thesis, Universit\'e Paris 6, 2014.

\bibitem[Panteleev and Kalachev(2021{\natexlab{a}})]{PK19}
Pavel Panteleev and Gleb Kalachev.
\newblock {Degenerate quantum LDPC codes with good finite length performance}.
\newblock \emph{Quantum}, 5:\penalty0 585, 2021{\natexlab{a}}.
\newblock \doi{10.22331/q-2021-11-22-585}.

\bibitem[Panteleev and Kalachev(2021{\natexlab{b}})]{PK20}
Pavel Panteleev and Gleb Kalachev.
\newblock {Quantum LDPC codes with almost linear minimum distance}.
\newblock \emph{IEEE Transactions on Information Theory}, 68\penalty0
  (1):\penalty0 213--229, 2021{\natexlab{b}}.
\newblock \doi{10.1109/TIT.2021.3119384}.

\bibitem[Quintavalle et~al.(2021)Quintavalle, Vasmer, Roffe, and
  Campbell]{QVR20}
Armanda~O. Quintavalle, Michael Vasmer, Joschka Roffe, and Earl~T. Campbell.
\newblock Single-shot error correction of three-dimensional homological product
  codes.
\newblock \emph{PRX Quantum}, 2:\penalty0 020340, 2021.
\newblock \doi{10.1103/PRXQuantum.2.020340}.

\bibitem[Serre(2002)]{serre2002matrices}
Denis Serre.
\newblock \emph{Matrices: Theory and Applications}, volume 216.
\newblock Springer Science \& Business Media, 2002.

\bibitem[Sipser and Spielman(1996)]{SS96}
Michael Sipser and Daniel~A Spielman.
\newblock Expander codes.
\newblock \emph{IEEE transactions on Information Theory}, 42\penalty0
  (6):\penalty0 1710--1722, 1996.
\newblock \doi{10.1109/18.556667}.

\bibitem[Steane(1996{\natexlab{a}})]{ste96b}
Andrew Steane.
\newblock Multiple-particle interference and quantum error correction.
\newblock \emph{Proc. R. Soc. Lond. A}, 452\penalty0 (1954):\penalty0
  2551--2577, 1996{\natexlab{a}}.
\newblock \doi{10.1098/rspa.1996.0136}.

\bibitem[Steane(1996{\natexlab{b}})]{ste96}
Andrew~M Steane.
\newblock Error correcting codes in quantum theory.
\newblock \emph{Physical Review Letters}, 77\penalty0 (5):\penalty0 793,
  1996{\natexlab{b}}.
\newblock \doi{10.1103/PhysRevLett.77.793}.

\bibitem[Terhal(2015)]{ter15}
Barbara~M Terhal.
\newblock Quantum error correction for quantum memories.
\newblock \emph{Reviews of Modern Physics}, 87\penalty0 (2):\penalty0 307,
  2015.

\bibitem[Tillich and Z{\'e}mor(2013)]{TZ13}
Jean-Pierre Tillich and Gilles Z{\'e}mor.
\newblock {Quantum LDPC codes with positive rate and minimum distance
  proportional to the square root of the blocklength}.
\newblock \emph{IEEE Transactions on Information Theory}, 60\penalty0
  (2):\penalty0 1193--1202, 2013.
\newblock \doi{10.1109/TIT.2013.2292061}.

\bibitem[Vuillot and Breuckmann(2022)]{VB19}
Christophe Vuillot and Nikolas~P. Breuckmann.
\newblock Quantum pin codes.
\newblock \emph{IEEE Transactions on Information Theory}, 2022.
\newblock \doi{10.1109/TIT.2022.3170846}.

\bibitem[Wu and Dawson(1998)]{wu1998existence}
Chuan-Kun Wu and Ed~Dawson.
\newblock Existence of generalized inverse of linear transformations over
  finite fields.
\newblock \emph{Finite Fields and Their Applications}, 4\penalty0 (4):\penalty0
  307--315, 1998.
\newblock \doi{10.1006/ffta.1998.0215}.

\bibitem[Zeng and Pryadko(2019)]{ZP19}
Weilei Zeng and Leonid~P. Pryadko.
\newblock Higher-dimensional quantum hypergraph-product codes with finite
  rates.
\newblock \emph{Phys. Rev. Lett.}, 122:\penalty0 230501, Jun 2019.
\newblock \doi{10.1103/PhysRevLett.122.230501}.
\newblock URL \url{https://link.aps.org/doi/10.1103/PhysRevLett.122.230501}.

\bibitem[Zeng and Pryadko(2020)]{ZP20}
Weilei Zeng and Leonid~P. Pryadko.
\newblock Minimal distances for certain quantum product codes and tensor
  products of chain complexes.
\newblock \emph{Phys. Rev. A}, 102:\penalty0 062402, 2020.
\newblock \doi{10.1103/PhysRevA.102.062402}.

\end{thebibliography}

\end{document}